\documentclass{jams-l}
\usepackage[title]{appendix}
\usepackage{amssymb}
\usepackage{amsmath}
\usepackage{caption}
\usepackage{amsthm}
\usepackage{array}
\usepackage{subfig}
\usepackage{graphicx}
\usepackage{enumitem}
\usepackage{hyperref}
\usepackage[numbers]{natbib}   
\usepackage{lipsum}
\usepackage{epigraph}
\usepackage{ MnSymbol }
\usepackage{natbib}
\usepackage{dsfont}
\usepackage{geometry}
\theoremstyle{plain}
\newtheorem{thm}{Theorem}[subsection] 
\theoremstyle{plain}
\newtheorem{defn}[thm]{Definition} 
\newtheorem{prop}[thm]{Proposition}
\newtheorem{lemma}[thm]{Lemma}
\theoremstyle{definition}
\newtheorem{exmp}[thm]{Example} 
\makeatletter
\def\th@plain{%
  \thm@notefont{}
  \itshape 
}
\def\th@definition{%
  \thm@notefont{}
  \normalfont 
}
\makeatother

\begin{document}

\title{Kontsevich's deformation quantization: from Dirac to multiple zeta values}

\author{Oisin Kim}

\newcommand\blfootnote[1]{%
  \begingroup
  \renewcommand\thefootnote{}\footnote{#1}%
  \addtocounter{footnote}{-1}%
  \endgroup
}
\blfootnote{Date: \today.}

\begin{abstract}

One way of reconciling classical and quantum mechanics is deformation quantization, which involves deforming the commutative algebra of functions on a Poisson manifold to a non-commutative, associative algebra, reminiscent of the space of quantum observables. This depends on a formal parameter $\hbar$, so that the original pointwise product is recovered when $\hbar=0$. In \cite{Kontsevich}, Kontsevich showed that a deformation quantization exists for every Poisson manifold. He furthermore gave a simple, combinatorial formula for producing a quantization of any Poisson structure on $\mathbb{R}^n$. 

The primary aim of this essay, largely drawn from the author's MMath dissertation at Oxford, is to present and explain Kontsevich's results. Starting with the motivation, we discuss how the problem is solved by situating it in a richer mathematical structure, performing a few original calculations along the way. We hope to communicate a sense of the strange links between this subject and seemingly distant areas of mathematics, and also to describe some of the contemporary research in the field. To these ends, we consider a recent paper, \cite{Panzer}, which connects deformation quantization to multiple zeta values.  
\end{abstract}

\maketitle
\tableofcontents

\section{Introduction}
\epigraph{“I have no way of knowing whether the events that I am about to narrate are effects or causes.”}{\textit{Jorge Luis Borges \\ Ficciones}}

\numberwithin{equation}{subsection} 
The aim of mechanics is to describe a physical system. In modern physics, there are two main ways of dealing with this. Classical mechanics, first formalised by Newton in the $17^{\text{th}}$ century, gives a deterministic description of physics on the macroscopic scale. Quantum mechanics, discovered in the early $20^{\text{th}}$ century, gives a microscopic description of physics that is often interpreted probabilistically. To motivate the subsequent discussion, we will briefly recap some key facts about these theories.
\subsection{Background}
One of the two standard descriptions of classical mechanics on $\mathbb{R}^n$ is the Hamiltonian formalism. Here the phase space is parametrized by canonical coordinates, consisting of generalized coordinates $\textbf{q}:=(q_1,...q_n)$ and conjugate generalized momenta $\textbf{p}:=(p_1,...,p_n)$. The system's energy is represented by a Hamiltonian function $H(\textbf{q},\textbf{p})$; the time evolution is described by Hamilton's equations. Given functions $f(\textbf{q},\textbf{p})$ and $g(\textbf{q},\textbf{p})$, the Poisson bracket is defined in this setting as:
\begin{equation*}\{f,g\} := \frac{\partial f}{\partial q^i}\frac{\partial g}{\partial p_i} - \frac{\partial f}{\partial p_i}\frac{\partial g}{\partial q^i},
\end{equation*}
where we have adopted Einstein notation for repeated indices, as we will for the remainder of the essay. With this construction, the time evolution of $f(\textbf{q},\textbf{p},t)$ can be easily expressed in the form: 
\begin{equation*} \frac{\text{d}f}{\text{d}t} = \{f,H\}+\frac{\partial f}{\partial t}.
\end{equation*}
This setup can then be extended to a general phase space. In classical mechanics, the states often form a smooth manifold $M$, typically equipped with a Poisson bracket, making it a Poisson structure. Here, the observables are the real-valued smooth functions $C^{\infty}(M)$; the algebra of observables is commutative. The motivating example is $M=T^{\ast}N$, the cotangent bundle of the space of configurations $N$. This has a coordinate-free Poisson bracket, making it into a Poisson manifold.

On the other hand, in quantum mechanics, the states are rays in a Hilbert space. The observables are then self-adjoint operators acting on the states. The evolution of the quantum system is governed by the Schr{\"o}dinger equation:
\begin{equation*}i\hbar \frac{\text{d}}{\text{d}t}\psi(t)=H\psi(t),
\end{equation*}
where $H$ denotes the quantum observable Hamiltonian. Clearly, quantum observables do not necessarily commute. 

Following the formation of quantum theory, many physicists and mathematicians sought to reconcile these systems by \emph{quantizing} classical mechanics. This may seem logically strange, as in principle we should derive classical behaviours from quantum theory. Unfortunately, our only existing methods for modelling quantum systems start with classical descriptions. 
\subsection{Quantization}
A natural idea for quantizing a classical system is to convert its observables into operators on a Hilbert space. We call this Hilbert space quantization, where a linear map $Q$ maps functions in $C^{\infty}(M)$ to self-adjoint operators. However, it is clear that classical and quantum observables are incompatible in many ways. For example, we have the problem of operator ordering: $Q(f)$ and $Q(g)$ do not need to commute, but in the classical case $fg = gf$ always. This makes our mapping inconsistent.

A natural idea, proposed by P.A.M. Dirac in \cite{Dirac}, is that $Q$ satisfies the condition:
\begin{equation*}Q(\{f,g\})=\frac{1}{i\hbar}[Q(f),Q(g)].
\end{equation*}
When $f$ and $g$ are canonical coordinates and momenta, we see that this becomes the Heisenberg commutation relation. Here Dirac considered $\hbar$ as a numerical value, Planck's constant, rather than as a formal variable in $\mathbb{R}[[\hbar]]$, a distinction that will be important later in our exposition. It was famously shown by H.J. Groenewold in \cite{Groenewold} that a map satisfying the above condition, which maps arbitrary polynomials in $x$ and $p$ to operator polynomials in $X$ and $P$, does not exist (if we also send the constant function $1$ to the identity).

There are various ways of overcoming this problem, but we must discard some of the assumptions of Groenewold's theorem, for example by restricting the space of functions to be quantized. One standard way of doing this is geometric quantization \cite{Woodhouse}, a type of Hilbert space quantization. A concrete example is the two-sphere:
\begin{equation*} x^2+y^2+z^2=j^2,
\end{equation*}
viewed as a symplectic manifold. In this case, even the simple observables $x,y,z$ can only be quantized for certain $j$. This is because: (i) any geometric quantization maps into a finite-dimensional Hilbert space, and (ii) the two-sphere's Poisson structure induces on the linear span of $x,y,z$ a Lie algebra structure isomorphic to the Lie algebra of angular momentum operators \cite{Hawkins}. In fact, here this implies that $j\in \frac{1}{2}\mathbb{N}$.

\subsection{Deformation quantization}
An alternative approach, and the subject of this essay, is \emph{deformation quantization}. The idea is to deform the commutative algebra of functions in $C^{\infty}(M)$ to a non-commutative, associative algebra. The deformation depends on a formal parameter $\hbar$, so that our new product reduces to the original commutative pointwise product when $\hbar=0$. In \cite{Kontsevich}, Kontsevich showed that a deformed \emph{star product} exists for every finite-dimensional Poisson manifold. He furthermore gave an explicit graphical formula for constructing a star product when the Poisson manifold is an open subset of $\mathbb{R}^d$. 

Note the difference with geometric quantization, or Hilbert space quantization in general: we are just deforming the algebra of the observables, without actually introducing a Hilbert space to act on. In this sense, deformation quantization is an abstract, mathematical affair, with ambiguous physical applicability. Everything is formal; we do not even worry about convergence. Returning to the two-sphere, Kontsevich has shown a deformation quantization exists for any $j$, whilst we already know a geometric quantization only exists for $j\in  \frac{1}{2} \mathbb{N}$.
\subsection{Outline of work}
In this essay, we will present and explain Kontsevich's results. Along the way they touch on many different (and beautiful) areas of mathematics; we aim to give the reader some appreciation for the surprising connections that emerge. 

Kontsevich's main result comes from embedding star products and Poisson structures into a richer mathematical structure, as is often the case in mathematics. We will see that Poisson structures are just deformations of the zero Poisson structure, and that any star product is obviously a deformation of function multiplication:
\begin{equation*}
f \star g = fg + B_1(f,g)\hbar + B_2(f,g){\hbar}^2+...
\end{equation*}
Since we are working with deformation problems, we will need to introduce the notion of a differential graded Lie algebra, or DGLA: specifically, we need a map between the DGLAs of multivector fields and multidifferential operators. A first attempt will be the HKR map of Hochschild, Konstant and Rosenberg \cite{Hochschild}. However, it does not respect the multiplicative structures on either side, a problem because our deformations correspond to solutions of the Maurer-Cartan equation:
\begin{equation*}
\textnormal{d}a+\frac{1}{2}[a,a]=0.
\end{equation*}
Ultimately, we will see that Kontsevich solved this issue with a different kind of mapping, known as an $L_{\infty}$-quasi-isomorphism:
\begin{equation*}
U: T_{poly}(M) \rightarrow D_{poly}(M).
\end{equation*}
We then examine Kontsevich's graphical formula in more depth. For associativity there must be weight coefficients in front of the star product's operators. Remarkably, these real numbers are expressible as rational linear combinations of multiple zeta values, a fact proven constructively in \cite{Panzer}. We will discuss this recent work, as well as other mysterious relationships between this (physically motivated) subject and number theory.

We perform a few original calculations. Firstly, in Section \ref{amend15}, we show Proposition \ref{prop1} using elementary techniques and our method. Using this, we give a simple proof to the degree one part of the HKR theorem in Proposition \ref{prop5}. This is the most relevant case for the purposes of deformation quantization, because Kontsevich's broader formality theorem only modifies the bijection between Maurer-Cartan sets at $O(\hbar^2)$; to our knowledge, an elementary proof did not previously exist in the literature. Finally, Appendix \ref{amend14} contains an original proof of Lemma \ref{lem2}, utilising the algorithm of \cite{Panzer} in an elementary setting (with heavy input from our supervisor). All figures are produced originally.

\subsection{Acknowledgements}
Firstly, I would like to sincerely thank my supervisor, Dr Erik Panzer, who proposed the topic of the dissertation this work is largely drawn from. His help and direction were invaluable throughout the writing of this essay, and were all the more appreciated during an unusual year. Most corrections were based on his scrupulously careful reading and insightful commentary, both of which went far beyond the standard call of marking duty. 

I would also like to thank those closest to me during the writing of this essay: my father, for providing intellectual guidance; my mother, for her emotional support and encouragement; and finally my girlfriend, Sophie, and brother, Niall, for both generally putting up with me.
\section{Basics}
Here we introduce the basic algebra necessary to define deformation quantization. This section's material is drawn in large part from \cite{Cattaneo}, \cite{Esposito}, and \cite{Kontsevich}. 
\subsection{Poisson structures}

 We want to formulate what it means for classical mechanics to exist on a smooth manifold $M$, not just $\mathbb{R}^n$. We begin with:
\begin{defn}[Poisson algebra]
A Poisson algebra is a vector space $V$ over a field $\mathbb{K}$ equipped with two $\mathbb{K}$-bilinear products. Firstly, a commutative product that makes $V$ into an commutative and associative $\mathbb{K}$-algebra. Secondly, a Poisson bracket $\{.,.\}:V\times V \rightarrow V$, that makes $V$ into a Lie algebra, and furthermore satisfies the Leibniz rule:
\begin{equation*} \{x,yz\}=y\{x,z\}+\{x,y\}z.\end{equation*}
\end{defn}
We saw in the introduction the standard Poisson bracket for $\mathbb{R}^n$; this generalises for arbitrary $M$ with:
 \begin{defn}[Poisson manifold]
A Poisson manifold is a smooth manifold M, equipped with a product $\{.,.\}:C^{\infty}(M)\times C^{\infty}(M) \rightarrow C^{\infty}(M)$ that makes $(C^{\infty}(M),\{.,.\})$ into a Poisson algebra.
\end{defn}
Here, the pointwise multiplication on $C^{\infty}(M)$ is the commutative product. To further understand this construction, we consider bivector fields. We want the Poisson bracket to satisfy the Leibniz rule in each component separately, in local coordinates $\{x^1,...,x^n\}$. So the map $f\rightarrow \{f,g\}$ is a derivation, with the same true for the other argument. It follows that there exists a unique $ \alpha \in \Gamma(\wedge^{2}TM) \; \text{with} \; \{f,g\} = \alpha(\text{d}f,\text{d}g) = (\text{d}f \otimes \text{d}g) (\alpha)$,
where $\Gamma(\wedge^{2}TX)$ denotes the space of skew-symmetric bivector fields. We can write:
\begin{equation*}  \label{ref3} \alpha=\alpha^{ij}\partial_i\otimes \partial_j,\end{equation*}
with $\alpha^{ij}=\{x^i,x^j\}$. 

Of course, the fact that $\{.,.\}$ satisfies the Jacobi identity should also impose a condition on the $\alpha^{ij}$. This is most easily formulated in terms of the \emph{Schouten-Nijenhuis bracket}, which acts between multivector fields. Let $\mathfrak{X}^k(M)=\Gamma({\wedge}^k TM)$, with $\mathfrak{X}^0(M)=C^{\infty}(M)$: we will use this notation for the remainder of this essay. On $\mathfrak{X}^1(M)$ there is the usual Lie bracket, which defines the \emph{Lie derivative}:
\begin{equation} \label{ref2}
    \mathcal{L}_{X}(Y):=[X,Y]=XY-YX.
\end{equation}
\begin{defn}[Schouten-Nijenhuis bracket]
The Lie bracket (\ref{ref2}) can be extended to  an operation: $[ .,.]_S: \mathfrak{X}^a(M) \otimes \mathfrak{X}^b(M) \rightarrow \mathfrak{X}^{a+b-1}(M)$, the Schouten-Nijenhuis bracket, defined on homogeneous elements by:
\begin{multline} \label{ref15}
  [X_1\wedge...\wedge X_a,Y_1\wedge...\wedge Y_b]_S := \\ \sum_{i=1}^{a}\sum_{j=1}^{b}(-1)^{i+j}[X_i,Y_j]\wedge X_1  \wedge ... \wedge X_{i-1} \wedge X_{i+1} \wedge ... \wedge X_a \wedge Y_1 \wedge ... \wedge Y_{j-1} \wedge Y_{j+1} \wedge ... \wedge Y_b,
\end{multline}
and extended linearly to all multivector fields.
\end{defn}
Here the wedge product denotes the anti-symmetrization of the tensor product. When $a=1$ this defines the Lie derivative $\mathcal{L}_{X_1}(Y_1 \wedge ... \wedge Y_b)$ of multivector fields. We will later naturally extend these definitions to formal multivector fields without comment. Armed with $[.,.]_S$, we have the easy condition:
\begin{prop}
The bivector $\alpha$ satisfies the Jacobi identity if and only if:
\begin{equation} \label{ref4}
[\alpha,\alpha]_S = 0.
\end{equation}
So this is also the condition necessary for it to define a Poisson structure.
\end{prop}
To prove this, one can show that the Jacobi identity and (\ref{ref4}) are both equivalent to a condition in $\{x^i\}$, but the full calculation is not illuminating (see \cite{Esposito}  for details). A final note about this bracket: it makes the space of multivector fields into a differential graded Lie algebra, with a shift of grading (this will be discussed further when we address the formality theorem). 

We now consider some simple examples. Firstly, there are constant Poisson structures on $\mathbb{R}^n$, when the  $\alpha^{ij}$ are constant functions with $\alpha_{ij}=-\alpha_{ji}$; these will later correspond to the Moyal product.
\begin{exmp}[Poisson structure on $\mathbb{R}^2$]
On $\mathbb{R}^2$ it is easily checked that all smooth Poisson structures are of the form:
\begin{equation*}
\{f,g\}=q(x,y)\left(\frac{\partial f}{\partial x} \frac{\partial g}{\partial y} - \frac{\partial f}{\partial y}\frac{\partial g}{\partial x}\right),
\end{equation*}
where $q: \mathbb{R}^2 \rightarrow \mathbb{R}$ is an arbitrary smooth function. 
\end{exmp}
For a smooth Poisson structure on $\mathbb{R}^3$, we can think of $\mathbb{R}^3$ as the dual of the Lie algebra $so(3)$. In fact, given an arbitrary real, finite-dimensional Lie algebra $\mathfrak{g}$, we can easily construct a Poisson structure on the dual $\mathfrak{g}^*$ by identifying $\mathfrak{g}$ with its double dual $\mathfrak{g}^{**}$ \cite{Laurent}. This is because the Lie bracket is an element of $(\Lambda^2 \mathfrak{g}^*) \otimes \mathfrak{g}$. 
\begin{exmp}[Poisson structure on $\mathbb{R}^3$]
Using this recipe for $\mathbb{R}^3$, we get a Poisson bracket on $C^{\infty}(\mathbb{R}^3) \times C^{\infty}(\mathbb{R}^3)$, defined by:
\begin{equation*}
\{.,.\} := x  \frac{\partial}{\partial y} \wedge \frac{\partial}{\partial z} + y \frac{\partial }{\partial z} \wedge \frac{\partial }{\partial x} + z \frac{\partial }{\partial x} \wedge \frac{\partial }{\partial y}.
\end{equation*}
\end{exmp}
\noindent We have now seen some easy Poisson structures, corresponding to the classical systems we aim to deform. The Poisson bracket will play an important role in bridging classical  and quantum mechanics; it is already introducing some non-commutativity into the classical realm.
\subsection{Star products}
Recall that we aim to quantize our system by deforming the pointwise multiplication on $C^{\infty}(M)$ into a non-commutative, associative product. This should depend on a formal parameter $\hbar$, so that $\star$ reduces to pointwise multiplication when $\hbar=0$.  To preserve a sense of the quantum mechanical origins of the theory, we use $\hbar$ as our formal variable, but this is not Planck's constant. If we ever wished to actually interpret the results, we would use $\hbar \rightarrow i \hbar$.
\begin{defn}[Star product]
A star product in $C^{\infty}(M)$ is an associative, $\mathbb{R}[\![\hbar ] \!]$ bilinear product:
\begin{equation*} \star: C^{\infty}(M)[\![\hbar ] \!] \times C^{\infty}(M)[\![\hbar ] \!] \rightarrow  C^{\infty}(M)[\![\hbar ] \!], \end{equation*}
acting on $f,g\in C^{\infty}(M) \subset C^{\infty}(M)[\![\hbar ] \!] $ as: 
\begin{equation}
\label{ref14}
f \star g = \sum_{i=0}^{\infty} B_i(f,g) {\hbar}^i,
\end{equation}
and extended $\mathbb{R}[\![\hbar ] \!]$ linearly. The $B_i$ are bidifferential operators and we require also that $B_0(f,g)=fg$.
\end{defn}
If $\hbar$ was a number, we would hope that the series on the right of (\ref{ref14}) would converge and define a function on M. To avoid worrying about this we just map into $C^{\infty}(M)[\![\hbar ] \!]$, where $\hbar$ is now regarded as a formal variable. Then for associativity we clearly need to take star products between formal power series also, so $\star$ must act on
 $ C^{\infty}(M)[\![\hbar ] \!] \times C^{\infty}(M)[\![\hbar ] \!]$.

We can now decompose $B_1$ into symmetric and anti-symmetric parts as:
\begin{equation*} \label{ref8}
B_1^-=\frac{1}{2}(B_1(f,g)-B_1(g,f)), \; \; 
B_1^+=\frac{1}{2}(B_1(f,g)+B_1(g,f)).
\end{equation*}
\begin{lemma}
Given a star product $f\star g= fg +B_1(f,g)\hbar+B_2(f,g){\hbar}^2...$, $B_1^-$ defines a Poisson structure on $M$. 
\end{lemma}
\begin{proof}
Associativity of $\star$ is equivalent to the condition:
\begin{equation}
\label{ref6}
\sum_{i+j=n} B_i(B_j(f,g),h) = \sum_{i+j=n} B_i(f,B_j(g,h)), \; \; \forall f,g,h \in C^{\infty}(M).
\end{equation}
For $n=1$ this implies the relation:
\begin{equation}
\label{ref7}
fB_1(g,h)-B_1(fg,h)+B_1(f,gh)-B_1(f,g)h=0,
\end{equation}
with the same identity for permutations of $f,g$ and $h$ (we will call this the cocycle condition for reasons that will become clear later in the exposition). Now, adding (\ref{ref7}) to one cyclic and one non-cyclic permutation implies that $B_1^-$ satisfies the Leibniz rule. If we take (\ref{ref6}) for $n=2$:
\begin{equation*}
B_2(fg,h)-B_2(f,gh)+B_1(B_1(f,g),h)-B_1(f,B_1(g,h))+B_2(f,g)h-fB_2(g,h)=0,
\end{equation*}
add the two cyclic permutations, then subtract the three non-cyclic permutations, we get exactly that $B_1^-$ satisfies the Jacobi identity. 
\end{proof}

We can see a relationship emerging between Poisson structures and star products, but it needs to be made more mathematically precise. This leads to:
\begin{defn}[Deformation quantization] \label{amend13}
A deformation quantization of a Poisson manifold is a star product for which the induced bracket satisfies $B_1^-(.,.)=\frac{1}{2}\{.,.\}$, where $\{.,.\}$ denotes the original Poisson bracket.
\end{defn}
From this definition, it is immediate that the ordinary product on functions extended $\mathbb{R}[\![\hbar]\!]$ linearly to $C^{\infty}(M)[\![\hbar]\!]$, which is commutative, corresponds to the zero Poisson structure. Note that the factor of $\frac{1}{2}$ means that the commutator of the star product satisfies Dirac's original quantization condition, at order $\hbar$.

\section{Kontsevich's result} \label{amend15}
The material in the next two sections draws on \cite{Cattaneo}, \cite{Keller}, \cite{Esposito}, \cite{Gutt} and \cite{Kontsevich}.
\subsection{Formal star products}
 Let $J$ be the set of automorphisms of $C^{\infty}(M)[\![\hbar ] \!]$, taken as a $\mathbb{R}  [\![\hbar ] \!]$ module, of the following form on $f \in C^{\infty}(M)$:
\begin{equation*}
T(f) = \sum_{n=0}^{\infty}T_n(f){\hbar}^n. 
\end{equation*}
Here the $T_n$ are linear differential operators. We require that they vanish on constants and also that $T_0$ is the identity. Given $T$ in this form, we can easily solve for the inverse, $D=T^{-1}$, in the same form,
where the $D_n$ are inductively defined by:
\begin{equation*}
D_0=\mathds{1},  \; \; D_n=-\sum_{m=0}^{n-1} T_m D_{n-m}.
\end{equation*} The fact that $T_0$ is the identity implies that the shifting is only happening in the formal portions of $C^{\infty}(M)[\![\hbar]\!]$.  Now we have:
\begin{defn}[Equivalent star products] \label{ref24}
If there exists $T \in J$ such that:
\begin{equation*}
T(u\star v) = T(u)\star' T(v), \; \; \forall u,v \in C^{\infty}(M) [\![\hbar]\!],
\end{equation*}
then $\star$ and $\star'$ belong to the same star product equivalence class. We also say that they are gauge equivalent.
\end{defn}
We now state and prove two important results about gauge equivalence. Firstly:
\begin{lemma} \label{lem1}
The $B_1^-$ Poisson bracket depends only on the star product equivalence class.
\end{lemma}
\begin{proof}
\emph{(Taken from \cite{Esposito})}. Given an arbitrary star product $\star$, apply $T\in J$, obtaining a gauge equivalent star product $\star'$:
\begin{equation*} \label{ref37}
\sum_{i=0}^{\infty} B_i(f,g) {\hbar}^i \rightarrow_T
\sum_{i=0}^{\infty} {C_i} (f,g) {\hbar}^i.
\end{equation*}
Using the formal expansion of $T$, this translates to the following condition for $T_1$:
\begin{equation} \label{ref38}
B_1(f,g)+T_1(fg)=C_1(f,g)+T_1(f)g+fT_1(g).
\end{equation}
So $B_1(f,g)-C_1(f,g)$ is symmetric in $f$ and $g$ and does not contribute to the induced bracket.
\end{proof}
Secondly:
\begin{prop} \label{prop1}
Any star product is gauge equivalent to a star product such that $B_1^+=0$.
\end{prop}
Kontsevich provides a highly abstract proof for this statement when he proves his own version of the HKR theorem, the entry point to the formality theorem, but the machinery used goes beyond the scope of this essay. Going through the calculation in depth is useful, as the ideas will later generalise to the HKR discussion. The existing proofs we found in the literature were either abstract or very sketchy with the details. So now we provide our own. 
\begin{proof}
Starting with an arbitrary star product, apply $T\in J$:
\begin{equation*}
\sum_{i=0}^{\infty} B_i(f,g) {\hbar}^i \rightarrow_T
\sum_{i=0}^{\infty} {C_i} (f,g) {\hbar}^i.
\end{equation*}
Then, (\ref{ref38}) gives us an explicit form for $C_1(f,g)$:
\begin{equation} \label{amend4}
C_1(f,g)=B_1(f,g)-fT_1(g)+T_1(fg)-T_1(f)g.
\end{equation}
Let $(x_k)$ be local coordinates for $M$. 
Say we start with the bidifferential operator $B_1(f,g)$:
\begin{equation*}
B_1= \sum_{I,J} a_{I,J} \partial_I \otimes \partial_J,
\end{equation*}
where the $I,J$ are coordinate multi-indices, such that $a_\mathbf{0}$ is non-zero, i.e. there is a term with no derivatives. Firstly choosing $T_1(f)= a_\mathbf{0} f$, we have that $C_1(f,g)=B_1(f,g)-a_\mathbf{0} f$, so we can assume up to gauge equivalence that $B_1$ has no order zero term.

We now want to enforce:
\begin{equation} \label{amend5}
C_1(f,g)+C_1(g,f)=0, \; \; \forall f,g \in C^\infty(M)
\end{equation}
Using (\ref{amend4}), (\ref{amend5}) is now equivalent to the $T_1$ condition:
\begin{equation} \label{ref30}
T_1(fg)=fT_1(g)+T_1(f)g-B_1^+(f,g), \; \; \forall f,g  \in C^\infty(M).
\end{equation}
We will use this to define $T_1$ locally on smooth functions. 

Assume that $T_1$ is defined separately on monomials containing only $\{x_1,...,x_{n-1}\}$ and monomials containing only ${x_n}$, such that (\ref{ref30}) holds true in either case. Define $T_1$ on monomials of $\{x_1,...x_n\}$, expressible as $x_n^{k_1} p_1$, $p_1=p_1(x_1,...,x_{n-1})$, as:
\begin{equation} \label{ref31}
T(x_n^{k_1}p_1):=T(x_n^{k_1})p_1+x_n^{k_1}T(p_1)-B_1^+(x_n^{k_1},p_1).
\end{equation}
Take two arbitrary such monomials $x_n^{k_1}p_1$ and $x_n^{k_2}p_2$. We want to show that:
\begin{equation} \label{ref32}
T(x_n^{k_1} p_1 x_n^{k_2} p_2) = T(x_n^{k_1}p_1)x_n^{k_2} p_2+x_n^{k_1}p_1T(x_n^{k_2} p_2)-B_1^+(x_n^{k_1}p_1,x_n^{k_2} p_2).
\end{equation}
Now using our definition (\ref{ref31}), we have that:
\begin{equation} \label{ref33}
T_1(x_n^{k_1} p_1 x_n^{k_2} p_2) = T_1(x_n^{k_1+k_2})p_1p_2+x_n^{k_1+k_2}T_1(p_1 p_2)-B_1^+(x_n^{k_1+k_2},p_1p_2).
\end{equation}
Now reapplying (\ref{ref31}) to (\ref{ref32}) and (\ref{ref33}), we see that their equivalence can be simplified to:
\begin{multline*}
p_1p_2B_1^+(x_n^{k_1},x_n^{k_2})+x_n^{k_1+k_2}B_1^+(p_1,p_2)+B_1^+(x_n^{k_1+k_2},p_1p_2) \\
= B_1^+(x_n^{k_1},p_1)x_n^{k_2}p_2+x_n^{k_1}p_1B_1^+(x_n^{k_2},p_2)+B_1^+(x_n^{k_1}p_1,x_n^{k_2}p_2).
\end{multline*}
But then this just comes from the associativity of the star product for different bracketing of four terms, restricting operators to the symmetric part. 

By induction we need only define $T_1$ for monomials of the form $x_k^n$, again so that (\ref{ref30}) holds. Before doing this, recall that:
\begin{equation} \label{amend9}
B_1^+= \sum_{I,J} a_{I,J} \partial_{I} \otimes \partial_J.
\end{equation}
Now consider the cocycle condition (\ref{ref7}) for  $B_1$; noting that it still holds when the arguments are swapped, it must also hold  for the operator $B_1^+$. Imposing $h=1$ gives us:
\begin{equation} \label{amend8}
fB_1^+(g,1)-B_1^+(fg,1)=0, \; \; \forall f,g \in C^\infty(M),
\end{equation}
where we have cancelled terms by the symmetry of the operator.
Substituting (\ref{amend9}) into (\ref{amend8}), note that all terms except those with $\partial_J=\mathds{1}$ vanish, so that:
\begin{equation*}
f\sum_{I} a_{I,0}\partial_I(g)-\sum_{I}a_{I,0}\partial_{I}(fg)=0, \; \; \forall f,g \in C^\infty(M).
\end{equation*}
In other words, the operator $\sum_{I} a_{I,0} \partial_I$ is function linear, and hence is simply multiplication by $a_{0,0}$, but by gauge equivalence we have already chosen this to be $0$. Thus $B_1^+$ has no terms of the form $a_{I,0} \partial_I \otimes \mathds{1}$, and since  $a_{I,J}=a_{J,I}$, there are no terms of the form $a_{0,I}\mathds{1}\otimes \partial_{I}$  either.

We will now define $T_1$ on monomials of the form $x^n$. Take $T_1$ to vanish on constant functions and linear functions. We then define:
\begin{equation} \label{ref35}
T_1(x^n) := T_1(x^{n-1})x-B_1^+(x^{n-1},x), \; \;  n \geq 2.
\end{equation}
We want to show that:
\begin{equation} \label{ref34}
T_1(x^{n} x^{m})=T_1(x^{n})x^{m}+x^{n} T_1(x^{m})-B_1^+(x^{n},x^{m}), \; \; \forall n,m.
\end{equation}
Fix $n$ arbitrarily; use induction on $m$. Say $T_1$ is defined such that (\ref{ref34}) holds for all $m \leq k$. We have that:
\begin{equation*}
T_1(x^{n+k+1})=T_1(x^{n+k})x-B_1^+(x^{n+k},x).
\end{equation*}
This is equivalent to:
\begin{equation*}
T_1(x^{n})x^{k+1}+x^nT_1(x^{k+1})-B_1^+(x^{n},x^{k+1}),
\end{equation*}
if and only if the condition:
\begin{equation*}
xB_1^+(x^n,x^k)+B_1^+(x^{n+k},x)=B_1^+(x^n,x^{k+1})+x^nB_1^+(x,x^{k}),
\end{equation*}
is satisfied. But this is just simple associativity of the star product, again restricting to the symmetric part.

$T_1$ is now defined for local monomials, and by linearity for polynomials. The condition (\ref{ref30}) holds for polynomials $f$ and $g$ as both sides of the equality are bilinear in $(f,g)$. It is stated in \cite{Cattaneo} that this is enough to define a differential operator with the desired properties on all local smooth functions. Unconvinced with the reasoning, we provide our own explanation. 

It is sufficient to show that our $T_1$ is expressible as a differential operator on all monomials. Start with one variable, say $x_k$. Using (\ref{ref35}), we have that:
\begin{equation} \label{ref36}
T_1(x_k^n)=x_kT_1(x_k^{n-1})-B_1^+(x_k,x_k^{n-1})=...=-\sum_{i=1}^{n-1}x_k^{n-1-i}B_1^+(x_k,x_k^i).
\end{equation}
Recall the decomposition (\ref{amend9}) for the operator $B_1^+$. We have already shown, through the cocycle condition, that it does not contain terms of the form: $\alpha_j(\mathds{1}\otimes \partial_k^j +\partial_k^j \otimes \mathds{1})$. But then $B_1^+(x_k,x_k^i)$ can only be non-zero for a special case, namely, for terms of the form:  
\begin{equation} \label{amend11}
\beta_j(\partial_k \otimes \partial_k^j+\partial_k^j \otimes \partial_k),
\end{equation}
restricting to $1 \leq j  \leq n-1$. For $j=1$ (\ref{ref36}) becomes:
\begin{equation*}
-\sum_{i=1}^{n-1} 2 \beta_1 i x_k^{n-1-i} x_k^{i-1}= -\sum_{i=1}^{n-1} 2 \beta_1 i x_k^{n-2}= -{n(n-1)}\beta_1 x_k^{n-2}=-\beta_1(\partial_k)^2(x_k^n):= D_1^1(x_k^n),
\end{equation*}
whilst for $2\leq j \leq n-1$ it becomes:
\begin{multline*}
-\sum_{i=1}^{n-1} \beta_j x_k^{n-1-i} \frac{i!}{(i-j)!}x^{i-j}= -\beta_j \left( \sum_{i=j}^{n-1} \frac{i!}{(i-j)!} \right) x_k^{n-j-1} \\ = - \beta_j \frac{n(n-1)...(n-j)}{j+1} x_k^{n-j-1}=-\frac{\beta_j}{j+1}(\partial_{k})^{j+1}(x^n_k):=D_j(x_k^n),
\end{multline*}
where we have used a combinatorial identity for the sum of partial factorials to treat the summed term in brackets.  
Defining $D_{x_k}:=\sum_{j} D_j$, these neat simplifications give us an expression for $T_1$ on single coordinate monomials as a single differential operator. 

For a monomial of $x_k$ and $x_l$, $k \neq l$, say $x_k^i x_l^j$, we use the the condition (\ref{ref30}) and identical reasoning. We have:
\begin{multline*}
T_1(x_k^i x_l^j)=T_1(x_k^i)x_l^j+x_k^iT_1(x_l^j)-B_1^+(x_k^i,x_l^j)
=\frac{1}{2}(D_{x_k}+D_{x_l})(x_k^ix_l^j)+\sum_{\gamma, \delta} \alpha_{\gamma,\delta}(\partial_k)^\gamma\otimes(\partial_l)^\delta(x_k^i,x_l^j) 
\\= \frac{1}{2}(D_{x_k}+D_{x_l})(x_k^ix_l^j)+\sum_{\gamma, \delta} \alpha_{\gamma, \delta} (\partial_k)^\gamma(\partial_l)^\delta(x_k^ix_l^j):=D_{x_k, x_l}(x_k^i x_l^j).
\end{multline*}
Arguing inductively on variables we have an expression for $T_1$ as a differential operator $D_{x_1...x_n}$ on all variables with the required properties. 

 We have found a differential operator $D$ satisfying (\ref{ref30}) on all polynomials of $\{x_1,...,x_n\}$; we need to check that (\ref{ref30}) holds for all smooth functions $f$ and $g$. Because this an equality of functions, we only need to check at points, where both sides only depend on polynomials of sufficiently large Taylor expansions of $f$ and $g$, to which we apply our previous arguments. 
 
 More precisely, given a point $p \in M$ and $f, g\in C^\infty(M)$, consider the set of all polynomial functions $h$ such that $h$ agrees exactly with the first $n$ terms of the Taylor series of $f$ at $p$; let $f_n$ denote any arbitrary element of this set. As $D$ and $B_1^+$ are differential operators with some highest-order component, by Taylor's theorem, we know that there exists $k \in \mathbb{N}$ such that:
\begin{equation*}
D(f)(p) = D(f_n)(p), \; \; B_1^+(f,g)(p) = B_1^+(f_n,g_n)(p), \; \; \forall n \geq k.
\end{equation*}
We moreover know that there exists $m \geq k$ such that:
\begin{equation*}
D(f_n g_n)(p) = D((fg)_k)(p)=D(fg)(p), \; \; \forall n \geq m,
\end{equation*}
with $m$ chosen large enough for the Taylor series of $fg$ to be covered up to order $k$ by $f_n g_n$. Now we are done, using the fact that we have already proved (\ref{ref30}) for all polynomials; we have that:
\begin{multline*}
D(fg)(p) = D(f_m g_m)(p) \\ =f_m (p)D(g_m)(p)+ D(f_m)(p)g_m(p)-B_1^+(f_m,g_m)(p)= f(p)D(g)(p) +D(f)(p)g(p)-B_1^+(f,g)(p).
\end{multline*}
Finally, choose a coordinate covering for our manifold: \begin{equation*}
    M=\bigcup_i U_i.
\end{equation*} Now using the previous arguments for $\mathbb{R}^n$ locally, we obtain a collection of local differential operators $D_i$. Choose a partition of unity $\{\rho_i\}_{i \in I}$ subordinate to the cover $\bigcup_i U_i$, and define $D:= \sum_i \rho_i D_i $, noting that multiplication by $\rho_i$ makes $D_i$ into a global operation. Using that $\sum_i \rho_i =1$, we obtain the global relationship:
\begin{equation*}
D(fg)=fD(g)+D(f)g-B_1^+(f,g),
\end{equation*}
and as $D$ is locally a differential operator, it is a global differential operator with all the required properties.
\end{proof}

Because we are only considering star products up to gauge equivalence, Proposition \ref{prop1} implies that if we can quantize a Poisson structure with $B_1^-=\{.,.\}$, we can quantize with $B_1=\{.,.\}$. From here on we will just consider this.
\subsection{Formal Poisson structures}
We can now do something similar for Poisson structures:
\begin{defn}[Formal Poisson structures]
A formal Poisson structure is a formal power series $\pi_{\hbar}= \sum_{i=0}^{\infty} \hbar^i \pi_i$, where the $\pi_i$ are skew-symmetric bivector fields. Also the Schouten-Nijenhuis bracket of $\pi_{\hbar}$ with itself must vanish at each $\hbar$ order.
\end{defn}
Clearly, we do not need each $\pi_i$ to be a Poisson structure, only the leading order term $\pi_0$. For the actual structures we will deform, $\pi_0=0$. Then the formal power series can be thought of as deformations of the zero Poisson structure. When are formal Poisson structures equivalent? The recipe is as follows:
\begin{enumerate} 
\item Given a diffeomorphism $\phi: M \rightarrow N$, we have the pushforward $\phi_*$, where $\phi_* [X,Y]_S=[\phi_* X, \phi_* Y ]_S$. 
\item Because the $[.,.]_S$ bracket of a Poisson structure $\pi$ with itself is invariant of the pushforward, the group of diffeomorphisms of $M$ acts on Poisson structures, with $\pi_{\phi}=\phi_* \pi$. 
\item We can now generalise this action to formal Poisson structures.
\end{enumerate}
\begin{defn}[Equivalent formal Poisson structures] \label{ref22}
The formal Poisson structures $\pi_{\hbar}$ and $\pi_{\hbar}'$ are equivalent if there exists a formal diffeomorphism between them, so that:
\begin{equation*}
\pi_{\hbar}= \sum_{n=0}^{\infty} \frac{{\hbar}^n}{n!} (\mathcal{L}_X)^n \pi_{\hbar}^{'} =  \text{exp}(\hbar \mathcal{L}_X) \pi_{\hbar}'.
\end{equation*}
Here $X$ is a formal vector field lying in $\mathfrak{X}^1(M)[\![\hbar]\!]$ and $\mathcal{L}_X$ is the Lie derivative.
\end{defn}
\subsection{Kontsevich's theorem}
We now have everything necessary to formulate Kontsevich's result.
\begin{thm}[Kontsevich, \cite{Kontsevich}] \label{ref27}
Every Poisson manifold has a deformation quantization. Furthermore, the set of gauge equivalence classes of star products on a smooth manifold $M$ can be naturally identified with the set of equivalence classes of Poisson structures depending formally on $\hbar$:
\begin{equation*}
\alpha_{\hbar} = \alpha(\hbar) = \alpha_1 \hbar + \alpha_2 {\hbar}^2+... \in \Gamma (\wedge^2 TM) [\![\hbar]\!], \;  \; [\alpha,\alpha]_S = 0 \in \Gamma (\wedge^3 TM) [\![\hbar]\!],
\end{equation*}
modulo the group of formal diffeomorphisms.
\end{thm}
As an example, we have a deformation quantization of our constant Poisson structure on $\mathbb{R}^n$, namely:
\begin{exmp}[The Moyal product] 
\begin{equation} \label{ref13}
f \star g:= \sum_{n=0}^{\infty} \frac{{\hbar}^n}{n!} \left( \prod_{k=1}^{n} \alpha^{i_k j_k} \right) \left( \prod_{k=1}^{n} \partial_{i_k} f \right) \left( \prod_{k=1}^{n} \partial_{j_k} g \right).
\end{equation}
\end{exmp} 
This was known before Kontsevich, but we will see later how it is rederived from his formula. It is the simplest non-trivial example we can consider.

\section{The formality theorem}
The goal of the following sections is to state Kontsevich's formality theorem, via describing the richer mathematical structure referred to in the introduction.
\subsection{Differential graded Lie algebras}
Recall that a \emph{graded vector space} $V$ is just a vector space that breaks up into a direct sum of the form: $V=\bigoplus_{n \in \mathbb{Z}} V^n$. If $v \in V^n$, we say $\text{deg}(v)=\bar{v}=n$.  We now have:
\begin{defn}[Differential graded Lie algebra] \label{ref18}
A differential graded Lie algebra $(V, \textnormal{d}, [.,.])$ is a graded vector space $V=\bigoplus_{n \in \mathbb{Z}} V^n$, together with a linear map $\textnormal{d}:V \rightarrow V$, and a bilinear map $[.,.]: V^i \otimes V^j \rightarrow V^{i+j}$, satisfying all the following conditions for any homogeneous $x,y,z \in V$: \begin{enumerate}
    \item \textnormal{d}, the differential, makes $(V,\textnormal{d})$ into a  chain complex: $\textnormal{d}(V^n) \in V^{n+1}$ and $\textnormal{d}^2=0$.
    \item $[.,.]$ is graded skew-symmetric: $[x,y]=(-1)^{\bar{x}\bar{y}+1}[y,x]$.
    \item $[.,.]$ satisfies the graded  Jacobi identity: $[x,[y,z]]=[[x,y],z]+(-1)^{\bar{x} \bar{y}}[y,[x,z]]$.
    \item $[.,.]$ and \textnormal{d} satisfy the graded Leibniz rule: $\textnormal{d}[x,y]=[\textnormal{d}x,y]+(-1)^{\bar{x}}[x,\textnormal{d}y]$.
\end{enumerate}
\end{defn}
Simply put, a DGLA is a complex with a compatible graded Lie algebra structure. If there is no differential, but just a bracket satisfying the above rules, this is a \emph{graded Lie algebra}, or GLA. Any GLA can be made into a trivial DGLA using the differential d $=0$. We say that elements in the kernel of d are \emph{closed}, or \emph{cocycles}, and that elements in image of d are \emph{exact}, or \emph{coboundaries}.
\begin{defn}[Cohomology complex]
Given a DGLA $V$, there is the cohomology quotient group:
\begin{equation*}
H^i(V) := \left. \textnormal{Ker}(\textnormal{d}: V^i \rightarrow V^{i+1})  \middle/ \textnormal{Im}(\textnormal{d}:  V^{i-1} \rightarrow V^i ) \right..
\end{equation*}
Then the $H := \bigoplus_{n} H^i(V)$ is itself a DGLA. It has an obvious graded vector space structure; an inherited bracket defined on equivalences classes $|x|,|y|$ by $[|x|,|y|]_H$ $:= |[x,y]_V|$; and finally a differential \textnormal{d} just set to be $0$.
\end{defn}
 Now a \emph{DGLA morphism} is a linear map $\phi: V_1 \rightarrow V_2$ of degree $0$, i.e. $\phi(V_1^i) \subset V_2^i$, that commutes with the brackets and differentials. Obviously a DGLA morphism $\phi:V_1 \rightarrow V_2$ induces a morphism $H(\phi): H_1 \rightarrow H_2$, between the cohomologies. 
\begin{defn}[Quasi-isomorphism] \label{ref19}
A quasi-ismorphism of DGLAs is a morphism of DGLAs that induces an isomorphism between the cohomologies.
\end{defn}
It is important to note that this morphism of DGLAs is compatible with the brackets, of both the DGLAs and cohomologies. On the other hand, a \emph{quasi-isomorphism of complexes} only respects the differential. This distinction will soon be important.
\subsection{DGLA examples}
As stated in the introduction, deformations of mathematical structures are controlled by DGLAs \cite{Manetti}. Here we are concerned with deformations of function multiplication and the zero Poisson structure, so it is natural DGLAs arise. Kontsevich identified two key DGLAs associated to these deformation problems.
\begin{defn}[$\mathbf{T_{poly}(M)}$]
Let $T_{poly}^n(M):=\mathfrak{X}^{n+1}(M)$. Take the graded vector space:
\begin{equation*}
    T_{poly}(M) := \bigoplus_{n=-1}^{\infty} T_{poly}^n(M).
\end{equation*}
The Schouten-Nijenhuis bracket for homogeneous elements is already defined by (\ref{ref15}). $T_{poly}(M)$ is a DGLA with $[.,.]_S$ and the differential \textnormal{d} $=0$, the DGLA of multivector fields.
\end{defn}
Note that we need to shift the degree in the direct sum to respect the definition of the bracket in Definition (\ref{ref18}). The cohomology just coincides with the DGLA itself. For our second example we must first define the Hochschild DGLA corresponding to an associative unital algebra $A$ (over a ring $\mathbb{K}$); it is just the DGLA on the complex of multilinear maps from $A$ to itself.
\begin{defn}[Hochschild DGLA] \label{amend1}
Let $C^n(A):=\text{Hom}_{\mathbb{K}}(A^{\otimes(n+1)},A)$. Take the graded vector space:
\begin{equation*}
C:=\bigoplus_{n=-1}^{\infty}C^n.  
\end{equation*}
Define a composition between $\phi \in C^{m_1}(A)$ and $\psi \in C^{m_2}(A)$ as:
\begin{equation} \label{ref16}
(\phi \circ \psi)(f_0,....,f_{{m_1}+{m_2}}):= \sum_{i=0}^{m_1} (-1)^{(m_2)(i)}\phi(f_0,...,f_{i-1},\psi(f_i,...,f_{i+m_2}),f_{i+m_2+1},...,f_{m_1+m_2}),
\end{equation}
for arbitrary $f_i, 0 \leq i \leq m_1+m_2+1$, elements of $A$. The Gerstenhaber bracket on homogeneous elements: $[.,.]_G:C^{m_1}\otimes C^{m_2} \rightarrow C^{m_1+m_2}$, is given by:
\begin{equation*}
[\phi,\psi]_G:=\phi \circ \psi - (-1)^{m_1 m_2} \psi \circ \phi.
\end{equation*}
We also have a differential \textnormal{d}: $C^n\rightarrow C^{n+1}$. An explicit formula can be written as:
\begin{multline} \label{ref17}
(-1)^n(\textnormal{d} \phi) (f_0,...,f_{n}) :=  \sum_{i=0}^{n-1}(-1)^{i+1}\phi(f_0,...,f_{i-1},f_if_{i+1},...,f_{n})\\ +f_0\phi(f_1,...,f_{n})+(-1)^{n-1}\phi(f_0,...,f_{n-1})f_{n}=(-\phi \circ \mu+(-1)^n \mu \circ  \phi )(f_0,...,f_{n})=-[\phi,\mu]_G(f_0,...,f_{n}),
\end{multline}
where $\mu \in C^1(A)$ denotes the ordinary product for $A$. Now $C$ is a DGLA with $[.,.]_G$ and the differential (\ref{ref17}).
\end{defn}
The fact this is a DGLA can be proved, for example in \cite{Cattaneo}, where they also provide a nice pictorial representation of the composition map (\ref{ref16}). There is again a shift in degree; we naturally interpret  $C^{-1}(A)$ as just $A$ itself. Note also that the differential has been simply expressed in terms of the Gerstenhaber bracket. We now have:
 \begin{lemma} \label{amend12}
  An element of $C^1(A)$ is an associative multiplication of the algebra $A$ if and only if the Gerstenhaber bracket with itself vanishes. 
 \end{lemma}
\begin{proof}
\emph{(Taken from \cite{Cattaneo}).} Given $\phi \in C^1(A)$, it is clear by Definition \ref{amend1} that $\phi$ is an associative multiplication of $A$ if and only if it satisfies the associativity condition:
\begin{equation} \label{amend2}
\phi(\phi(f,g),h)-\phi(f, \phi(g,h))=0, \; \; \forall f,g,h \in A.
\end{equation}
But now:
\begin{equation} \label{amend3}
[\phi, \phi]_G(f,g,h)=2\phi \circ \phi (f,g,h)=2( \phi(\phi(f,g),h)-\phi(f,\phi(g,h))).
\end{equation}
Clearly, (\ref{amend3}) vanishes $\forall f,g,h \in A$ if and only if condition (\ref{amend2}) holds.
\end{proof}

Let us now define $D_{poly}(M)$, simply restricting our last example to multidifferential operators. From here on $A$ denotes the specific algebra $C^{\infty}(M)$ and take $\mathbb{K}=\mathbb{R}$.
\begin{defn}[$\mathbf{D_{poly}(M)}$]
Let $D_{poly}^n(M)$, $n \geq  -1$, be the subset of $\text{Hom}(A^{\otimes(i+1)},A)$ containing only the multidifferential operators. An expression for the general form of these maps, taken from \cite{Kontsevich}, is:
\begin{equation}
f_0 \otimes... \otimes f_n \rightarrow \sum_{(I_0,...,I_n)}  C^{I_0,...,I_n}\partial_{I_0}(f_0)...\partial_{I_n}(f_n),
\end{equation}
where the sum is finite, the $I_k$ denote multi-indices, and the $f_k$, $C^{I_0,...I_n}$ are functions in the local coordinates.  Now take the graded vector space:
\begin{equation*}
D_{poly}(M):= \bigoplus_{n=-1}^{\infty} D_{poly}^n(M).
\end{equation*}
$D_{poly}(M)$ is closed under the Gerstenhaber bracket and the Hochschild differential. It is therefore a DGLA, with $[.,.]_G$ and \textnormal{d}, (\ref{ref17}).
\end{defn}
Why have we defined these DGLAs? Several relationships can already be stated. Given a star product $f\star g= fg+B_1(f,g)\hbar+...$, the bidifferential operator $B_1$ is in $D^1_{poly}(M)$. If we use a gauge transformation $T$ to transform $\star$ into $\star'$, as in (\ref{ref37}), the formula for $C_1$ gives:
\begin{equation} \label{ref41}
    C_1(f,g)-B_1(f,g)=T_1(fg)-T_1(f)g-fT_1(g)=\text{d}(-T_1)(f,g),
\end{equation}
an expression for gauge equivalence at order $\hbar$ in terms of the differential. So the gauge equivalence classes of closed bidifferential operators are just $H^1(D_{poly}(M))$.

Clearly Poisson structures on $M$ lie in $T_{poly}(M)$. Now $D_{poly}(M)[\![\hbar]\!]$ and $T_{poly}(M)[\![\hbar]\!]$ are also DGLAs, just tensoring with $\mathbb{R}[\![\hbar]\!]$ and extending operations $\mathbb{R}[\![\hbar]\!]$ linearly. Then formal Poisson structures that are deformations of the zero Poisson structure lie in $\hbar T_{poly}^1(M)[\![\hbar]\!]$ and star products correspond to elements of $\hbar D_{poly}^1(M)[\![\hbar]\!]$.

\subsection{The Maurer-Cartan equation}

We now explain our assertion in the introduction that Poisson structures and star products correspond to solutions of the Maurer-Cartan equation.
\begin{defn}[The Maurer-Cartan equation]
Given a DGLA $V$ there is the Maurer-Cartan equation:
\begin{equation*} 
\textnormal{d}a+\frac{1}{2}[a,a]_V=0, \; \;  a \in {V^1}.
\end{equation*}
For $V=D_{poly}(M)[\![\hbar]\!]$ or $T_{poly}(M)[\![\hbar]\!]$, the Maurer-Cartan set $MC(V)$ equals $\gamma \in \hbar V^1$ such that \textnormal{d}$\gamma+\frac{1}{2}[\gamma,\gamma]=0$.
\end{defn}
Consider the Maurer-Cartan sets of our two important DGLAs:
\begin{exmp}[$\mathbf{MC(T_{poly}(M))}$]
We know that $\pi$ is a Poisson structure $\iff$  $[\pi,\pi]_S=0$. Since for $\hbar T_{poly}(M)[\![\hbar]\!]$, $\textnormal{d}=0$, we see that formal Poisson structures with $\pi_0=0$ coincide with $MC(T_{poly}^1(M)[\![\hbar]\!])$. 
\end{exmp}
\begin{exmp}[$\mathbf{MC(D_{poly}(M))}$]
If $B \in \hbar D_{poly}^1(M)[\![\hbar]\!]$, by Lemma \ref{amend12} we know that the deformation $\mu+B$ is associative $\iff$ $[\mu+B,\mu+B]_G=0$. Now $[\mu,\mu]_G=0$, and as  $\bar \mu  = \bar B = 1$, we also have that $\text{d}(B)=[\mu,B]_G$ and $[\mu,B]_G=[B,\mu]_G$. Thus $B$ defines a star product $\mu+B \iff B \in MC(D_{poly}^1(M)[\![\hbar]\!])$.
\end{exmp}
\subsection{Gauge groups}
Finally, we also have gauge groups for our DGLAs. These further translate the previous discussions of gauge equivalence into our new framework. Take $G^{0}( T_{poly}(M)[\![\hbar]\!])$ to be the group of formal diffeomorphisms given in Definition \ref{ref22}; the group structure comes from the Baker-Campbell-Hausdorff formula:
\begin{equation*}
\textnormal{exp}(\hbar X) \textnormal{exp}(\hbar Y):= \textnormal{exp}\left( \hbar X + \hbar  Y = \frac{1}{2} \hbar [X,Y]+...\right).
\end{equation*}
Because the exponential of a derivation is an automorphism of the DGLA, the bracket and Maurer-Cartan set is preserved by $G^0(T_{poly}(M)[\![\hbar]\!])$. We can thus define the quotient space:
\begin{equation*} \label{ref23}
    \overline{MC}(T_{poly}(M)[\![\hbar]\!]) := \frac{MC(T_{poly}(M)[\![\hbar]\!])}{G^0( T_{poly}(M)[\![\hbar]\!])},
\end{equation*}
corresponding with formal Poisson structure equivalence classes. Similarly, define:
\begin{equation*}
G^0(D_{poly}(M)[\![\hbar]\!]) \\ := \{ \phi: D_{poly}(M)[\![\hbar]\!] \rightarrow D_{poly}(M)[\![\hbar]\!] : \phi = e^{\hbar[S,.]},\; \; S \in D_{poly}^0(M)[\![\hbar]\!]  \}.
\end{equation*}
Then taking the quotient set $\overline{MC}(D_{poly}(M)[\![\hbar]\!])$, elements in the same equivalence class correspond to gauge equivalent star products, in the sense of Definition \ref{ref24} (for a proof see \cite{Waldmann}). In fact, these constructions exist for any DGLA, but for simplicity we restrict to our examples.
\subsection{The HKR map}
From here on, we may sometimes be somewhat informal in our definitions, as anything defined for $T_{poly}(M)$ or $D_{poly}(M)$ can just be extended linearly to the tensors with $\mathbb{R}[\![\hbar]\!]$. Of course, our goal is ultimately to associate our two deformation problems. The first step towards achieving this is the HKR theorem.
\begin{thm}[Hochschild-Konstant-Rosenberg, \cite{Hochschild}]  \label{ref25}
Define the $A$-linear map $U_1: T_{poly}(M) \rightarrow D_{poly}(M)$ on homogeneous elements in $T_{poly}(M)$ as :
\begin{equation*} \label{ref29}
(U_1(X_0 \wedge ... \wedge X_n))(f_0,...,f_n) := \frac{1}{(n+1)!}\sum_{\sigma \in S_{n+1}} \text{sgn}(\sigma) \prod_{i=0}^{n} X_{\sigma(i)}(f_i),
\end{equation*}
with $f_0,...,f_n$ functions on $M$ and $\sigma$ denoting a permutation in the set of permutations of $n$ elements, $S_n$. Then this map is a quasi-isomorphism of complexes.
\end{thm}
This was stated originally in \cite{Hochschild} in a different setting. Kontsevich gives a proof for smooth manifolds, in \cite{Kontsevich}, but the proof is quite abstract; the relation to the original theorem is not even completely clear. Here we will only prove the HKR theorem at degrees zero and one.
\begin{prop}
\begin{equation*}
T_{poly}^0(M) \cong H^0(D_{poly}(M)).
\end{equation*}
\end{prop}
\begin{proof}
\begin{itemize}
    \item For a morphism of complexes, any $B \in T_{poly}^0(M)$ must be closed in $D_{poly}^0(M)$. Take $f,g$  arbitrary in $C^\infty(M)$. We have the cocycle condition:
    \begin{equation*}
    \textnormal{d}(B)(f,g)=-B(fg)+fB(g)+B(f)g=0.
    \end{equation*}
    This is just the statement that $B$ is a derivation, so we are done as $B$ is a vector field. 
    \item In this situation, injectivity is immediate. Surjectivity requires that every Hochschild cocycle is a vector field, but this follows by the argument above, noting that a  Hochschild cocycle is necessarily a derivation.
\end{itemize}
\end{proof}
\begin{prop} \label{prop5}
\begin{equation*}
T_{poly}^1(M)\cong H^1(D_{poly}(M)).
\end{equation*}
\end{prop}
Using Lemma \ref{lem1} and Proposition \ref{prop1} this is already nearly proven: we just fill in the remaining details ourselves here.
\begin{proof}
\begin{itemize}
\item For a morphism of complexes, any $B \in T_{poly}^1(M)$ must be closed in $D_{poly}^1(M)$. Take $f,g,h$ arbitrary in $C^{\infty}(M)$. We have the cocycle condition:
\begin{equation*}
\text{d}(B)(f,g,h)=fB(g,h)-B(fg,h)+B(f,gh)-B(f,g)h=0.
\end{equation*}
It is easy and sufficient to check this for $B=X\otimes Y- Y \otimes X$, $X$ and $Y$ arbitrary vector fields, by direct computation. 

\item For injectivity, any $B_1,B_2 \in T_{poly}^1(M)$, such that $B_1 \neq B_2$, must not be in the same cohomology class. This is just Lemma \ref{lem1}.

\item For surjectivity, every Hochschild cocycle must be cohomologous to a bivector field. Proposition \ref{prop1} together with the gauge transformation in terms of d, (\ref{ref41}), gives us a cohomologous skew-symmetric Hochschild cocycle. 
\end{itemize}
\ \\
If every skew-symmetric Hochschild cocycle $B$ is a bivector field, the proof is complete. It is sufficient to check that $B$ is a derivation in one variable. We want to show that:
\begin{equation*}
B(f,gh)=gB(f,h)+B(f,g)h, \; \; \forall f,g,h.
\end{equation*}
Recall the cocycle condition:
\begin{equation} \label{ref39}
fB(g,h)+B(f,gh)=B(fg,h)+B(f,g)h.
\end{equation}
Cyclically permuting $f,g,h$ then using skew-symmetry, we also have:
\begin{equation} \label{ref40}
gB(h,f)+B(g,hf)=B(g,h)f-B(f,gh).
\end{equation}
Subtracting (\ref{ref40}) from (\ref{ref39}) then gives the relation:
\begin{equation} \label{ref42}
2B(f,gh)=B(fg,h)+B(f,g)h-gB(h,f)-B(g,hf).
\end{equation}
Reapplying the cocycle condition implies that:
\begin{equation*}
B(fg,h)-B(g,fh)=gB(f,h)+B(f,g)h.
\end{equation*}
Substituting this into (\ref{ref42}) and using skew-symmetry, we are done.

\end{proof}

We have constructed a map from the DGLA controlling deformations of the zero Poisson structure to the DGLA controlling deformations of function multiplication. However, this is not the full story. We have just seen that formal Poisson structures and star products correspond to the Maurer-Cartan sets of $T_{poly}^1(M)[\![\hbar]\!]$ and $D_{poly}^1(M)[\![\hbar]\!]$: we therefore need a mapping that respects the DGLA structures. 

The HKR theorem only gives us a quasi-isomorphism of complexes. Because the multiplicative structure associated to $[.,.]_G$ is not mapped to $[.,.]_S$, it does not give a bijection of Maurer-Cartan sets. Ideally we would want a DGLA quasi-isomorphism, but a weaker map is enough: the $L_{\infty}$-quasi-morphism. Although the HKR theorem was insufficient, we will soon see that it was a necessary first step.

\subsection{L-infinity morphisms}

Recall the notion of a DGLA quasi-isomorphism $\phi_1:V_1 \rightarrow V_2$, given in Definition \ref{ref19}. The map acts from $V_1$ to $V_2$; we do not require that it has an inverse. 
Nevertheless, we may believe that an equivalence between $V_1$ and $V_2$ of the form:
\begin{equation*}
V_1 \rightarrow V_3 \leftarrow V_4 \rightarrow ... \leftarrow V_n \leftarrow V_2,    
\end{equation*}
where the arrows denote quasi-isomorphisms, should still have meaning. In fact, if $V_1$ and $V_2$ are related in the above way, there exists $V_3$  such that:
\begin{equation*}
V_1 \leftarrow V_3 \rightarrow V_2,
\end{equation*}
(see \cite{Keller}). If $V_1$ and $V_2$ satisfy this condition, we say that they are \emph{homotopy equivalent}. We will soon see that the condition of homotopy equivalence is closely related to our new type of mapping, the $L_{\infty}$-morphism. These are `weaker' versions of DGLA morphisms, where we relax the necessary conditions.

\begin{defn}[$\mathbf{L_{\infty}}$\textbf{-morphism}]
An $L_{\infty}$-morphism of DGLAs, $f: V_1\rightarrow V_2$ is a sequence of maps:
\begin{equation*}
f_n: {V_1}^{\otimes n} \rightarrow V_2, \; \;n\geq 1,
\end{equation*}
with each map homogeneous of degree $1-n$, satisfying the following conditions on arbitrary homogeneous elements $x_1,...,x_n$:
\begin{enumerate}
    \item The $f_n$ are graded anti-symmetric:
    \begin{equation*}
        f_n(x_1\otimes...\otimes x_i \otimes x_{i+1}\otimes ...\otimes x_n)= -(-1)^{\overline x_i \overline x_{i+1} }f_n(x_1\otimes...\otimes x_{i+1} \otimes x_i\otimes...\otimes x_n).
    \end{equation*}
   \item For $n\geq1$:
    \begin{multline} \label{ref20}
    \textnormal{d}f_n(x_1 \wedge x_2 \wedge ... \wedge x_n)-\sum_{i=1}^{n} \pm f_n(x_1 \wedge ... \wedge \textnormal{d}x_i \wedge ... \wedge x_n)\\
    =\frac{1}{2}\sum_{k,l\geq 1, k+l=n}\frac{1}{k!l!}\sum_{\sigma \in S_n} \pm [f_k(x_{\sigma_1}\wedge...\wedge x_{\sigma_k}),f_l(x_{\sigma_{k+1}},\wedge...\wedge x_{\sigma_n})]+
    \sum_{i<j} \pm f_{n-1}([x_i,x_j]\wedge x_1\wedge...\wedge x_n).
    \end{multline}
\end{enumerate}
\end{defn}
Here the notation of the wedge product is used because the value on the tensor product only depends on its  equivalence class in the graded wedge product. Homogeneous of degree $1-n$ means that for homogeneous elements we sum the degrees, then add $1-n$: $f_n$ maps into this grading. For $n=1$, (\ref{ref20}) implies $f_1$ is a \emph{morphism of complexes}: $f_1 \circ d = d \circ f_1$. For $n=2$, it implies that:
\begin{equation*} \label{ref21}
    f_1([x_1,x_2])-[f_1(x_1),f_1(x_2)]=\textnormal{d}(f_2(x_1,x_2))+f_2(\textnormal{d}(x_1), x_2)+(-1)^{\overline x_1}f_2(x_1, \textnormal{d}(x_2)),
\end{equation*}
for all homogeneous $x_1$ and $x_2$. The bracket structure is respected by $f_1$, up to the right hand side. In fact, if $x_1$ and $x_2$ are closed, the right hand side is exact and quotiented out by the image of d. So $f_1$ induces a morphism of cohomologies, considered themselves as DGLAs. 

$L_{\infty}$-morphisms generalise DGLA morphisms. We can now consider:
\begin{defn}[$L_{\infty}$-quasi-isomorphism]
An $L_{\infty}$-quasi-isomorphism is an $L_{\infty}$-morphism such that the first term, $f_1$, is a quasi-isomorphism of complexes. 
\end{defn}
Interestingly, this is the same as homotopy equivalence.
\begin{prop}
The following statements are equivalent:
\begin{enumerate}
    \item There exists $V_3$ such that: $V_1 \leftarrow V_3 \rightarrow V_2$.
    \item There exists an $L_{\infty}$-quasi-isomorphism from $V_1 \rightarrow V_2$.
\end{enumerate}
\end{prop}
For more discussion of this, see \cite{Keller}. At the cost of more theoretical machinery, we have a convenient single map, rather than the messy setup of multiple DGLA morphisms in different directions. 
\subsection{The formality theorem}
Before giving the main result, we state:
\begin{thm} \label{ref26}
Given an $L_{\infty}$-quasi isomorphism $f:T_{poly}(M)[\![\hbar]\!] \rightarrow D_{poly}(M)[\![\hbar]\!]$, the map:
\begin{equation*} \label{ref28}
x \rightarrow \sum_{n=1}^{\infty} \frac{1}{n!} f_n(x,...,x),
\end{equation*}
induces a bijection:
\begin{equation*}
  \overline{MC}(T_{poly}(M)[\![\hbar]\!]) \simeq \overline{MC}(D_{poly}(M)[\![\hbar]\!]).
  \end{equation*}
\end{thm}
This is in fact true for any two DGLAs, with a bijection induced between the two corresponding quotient sets previously alluded to. With this setup, we have everything necessary to understand Kontsevich's formality theorem.
\begin{thm}[Kontsevich, \cite{Kontsevich}]
There exists an $L_{\infty}$-quasi-isomorphism:
\begin{equation*}
U: T_{poly}(M)[\![\hbar]\!] \rightarrow D_{poly}(M)[\![\hbar]\!],
\end{equation*}
with the first component $U_1$ given by the HKR quasi-isomorphism of complexes of Theorem \ref{ref25}.
\end{thm}

The proof of these results go beyond the scope of this essay. Nevertheless, we have established sufficient machinery for the reader to understand that together they give us the natural identification of  Theorem \ref{ref27}. In particular, we have a deformation quantization of any given Poisson manifold, what we wanted all along. Note that if we want to quantize an ordinary Poisson structure $\alpha$, we input $\frac{\hbar \alpha}{2}$, again so that the star product's commutator has the right form at order $\hbar$.

The HKR map was indeed the first step in constructing the $L_{\infty}$-quasi-isomorphism. Moreover, this bijection is a deformation quantization, because $f_1$, the HKR map, sends $\hbar \alpha$ to itself, and all other components map into $\hbar$ order two and higher. In the final chapter, we will return to these ideas when we discuss the broader mathematical connections. One can formulate a notion of homotopy equivalence between the so-called \emph{stable} $L_{\infty}$-quasi-isomorphisms: then the real Grothendieck-Teichm{\"u}ller group $GT(\mathbb{R})$ acts on equivalence classes  \cite{Dolgushev}, \cite{Dolgushev1}.

\section{Kontsevich's formula}
Armed with Kontsevich's main theoretical result, we proceed to something more concrete. The $L_{\infty}$-quasi-isomorphism gives us a simple, explicit formula for finding a star product given a Poisson structure on a manifold that is an open subset of $\mathbb{R}^d$. This section's background material is drawn largely from \cite{Cattaneo}, \cite{Esposito} and \cite{Kontsevich}; the later discussion is based on \cite{Panzer}.
\subsection{Admissible graphs}
We are going to write our star product formula in terms of a family of \emph{admissible graphs}. These can be defined as follows.
\begin{defn}[Oriented graphs] 
An oriented graph $\Gamma$ is a pair $(V_{\Gamma}, E_{\Gamma},)$  such that $E_{\Gamma}$ is a subset of $V_{\Gamma} \times V_{\Gamma}$. 
\end{defn}
$V_{\Gamma}$ are the vertices of $\Gamma$ and $E_{\Gamma}$ are the edges of $\Gamma$; the edge $e=(v_1,v_2)$ starts at $v_1$ and ends at $v_2$. Formally, a labelled, oriented graph $\Gamma$ is admissible if: 
\begin{enumerate}
    \item $\Gamma$ has $n+2$ vertices, $\{1,...,n,L,R\}$, where $L,R$ are just symbols. 
    \item $\Gamma$ has $2n$ edges.
    \item The edges are labelled by $e_1^1,e_1^2,...,e_n^1,e_n^2$; the edges $e_k^1$, $e_k^2$ start at the vertex $k$. 
    \item The edges $e_k^1=(v_k,v_1)$ and $e_k^2=(v_k,v_2)$ are such that $v_1 \neq v_2$.
    \item There is no vertex $v$ such that the $(v,v) \in E_{\Gamma}$.
\end{enumerate}
We are not allowing infinite graphs or multiple edges. To restate the rules simply, two edges go out from the first $n$ vertices $\{1,...,n\}$; these edges must end at different vertices from each other. The vertices $\{L,R\}$ have no edges going out. Finally, we cannot have loops. Let  $G$ be the set of admissible graphs and $G_n$ be the set of admissible graphs with $n+2$ vertices. If $\Gamma  \in G_n$, we say $|\Gamma|=n$; note that $G_n$ contains $(n(n+1))^n$ elements. Now Kontsevich's formula is a sum over the graphs.
\begin{thm}[Kontsevich, \cite{Kontsevich}]
Given $\pi$, a Poisson structure on an open subset of $\mathbb{R}^d$, the formula:
\begin{equation} \label{ref9}
f \star g := \sum_{\Gamma \in G} \frac{{\hbar}^n}{n!} \omega_{\Gamma} B_{\Gamma, \pi}(f,g),
\end{equation}
defines an associative product. If we change coordinates, we obtain a gauge-equivalent star product.
\end{thm}
Collecting terms, the operator at order $h^n$ is given by:
\begin{equation*}
B_n(f,g) = \frac{1}{n!} \sum_{\Gamma \in G_n} {\omega_{\Gamma}} B_{\Gamma, \pi}(f,g).
\end{equation*}
Of course we need to explain the terms in the sum. Given a Poisson structure $\pi$, we assign a bidifferential operator $B_{\Gamma, \pi}$ to each graph $\Gamma$, depending on both $\Gamma$ and $\pi$. In front of this there is a weight, $\omega_{\Gamma} \in \mathbb{R}$, depending only on $\Gamma$. The star product of $f$ and $g$ is then given by the formal power series (\ref{ref9}).
\begin{figure} 
\includegraphics[width=12cm]{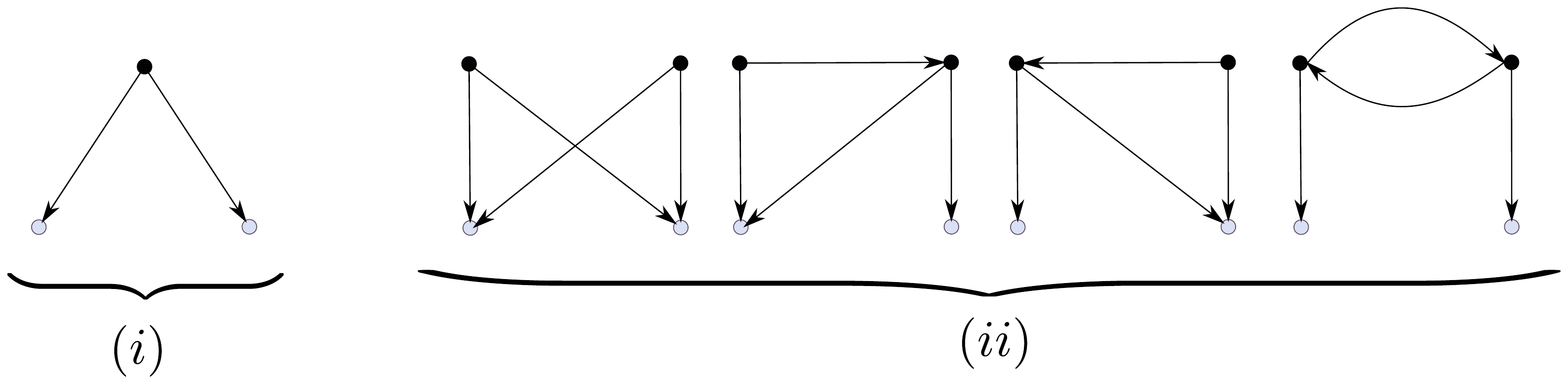}
\centering
\caption{Unlabelled admissible graphs for (i) $n=1$ and (ii) $n=2$.}
\centering
\label{fig:2}
\end{figure}
\subsection{Bidifferential operators}
Given a graph and Poisson bivector field, $\{f,g\}=\pi^{ij}\partial_i f \partial_j g$, how do we calculate $B_{\Gamma, \pi}(f,g)$? The recipe is as follows. 
\begin{enumerate}
    \item Use a map $I:E_{\Gamma} \rightarrow \{1,...,d\}: (e_1^1,e_1^2,...,e_n^1,e_n^2) \rightarrow (i_1,...,i_d)$, changing the labels of the edges.
    \item Assign $f$ to $L$ and $g$ to $R$. Assign the tensor $\pi^{I(e^1_k) I(e^2_k)}$ to the vertex $k$.
    \item For every edge labelled $i_j$, apply a partial derivative with respect to $i_j$ to the function or tensor associated with the vertex endpoint.
\item Following this setup, we can define a general formula, summing over maps $I$:
\end{enumerate}
\begin{equation*}
B_{\Gamma, \pi} := \sum_{I:E_{\Gamma} \rightarrow \{1,...,d\}} \left[ \prod_{k=1}^{n} \left( \prod_{e\in E_{\Gamma}, e=(*,k)} \partial_{I(e)} \right) \pi^{I(e_k^1) I(e_k^2)}\right]  \times \left( \prod_{e\in E_{\Gamma}, e=(*,L)} \partial_{I(e)} \right)f\left( \prod_{e\in E_{\Gamma}, e=(*,R)} \partial_{I(e)} \right)g.
\end{equation*}
In summary, each vertex is assigned a tensor or function, and edges define partial derivatives on their endpoints. Multiplying these together, we then sum over the edge's labels. To illustrate the rule, consider the graph in Figure \ref{fig:1}. The corresponding operator is:
\begin{equation*}
B_{\Gamma, \pi}(f,g)=\sum_{1 \leq i_1...i_6 \leq d} \pi^{i_5 i_6}  \partial_{i_5}\partial_{i_2}(\pi^{i_3 i_4})\partial_{i_3}(\pi^{i_1 i_2})
\partial_{i_1} \partial_{i_4}(f) \partial_{i_6} (g).
\end{equation*}
\begin{figure} 
\includegraphics[width=6cm]{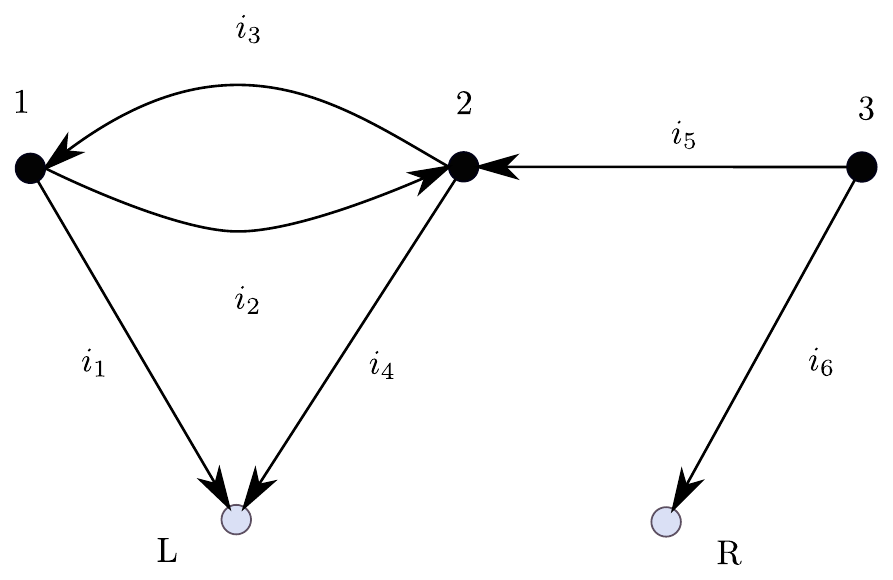}
\centering
\caption{An admissible graph of order 3.}
\label{fig:1}
\end{figure}
\subsection{Weight integrals}
For the weights $\omega_\Gamma$, as with the formula, we first give the definition, then explain the meaning of the terms. The weights $\omega_\Gamma$ are defined as integrals of differential $2n$-forms:
\begin{equation} \label{ref10}
    \omega_{\Gamma}:= \int_{{\mathcal{H}}_n} \text{d}\phi_{e_1} \wedge ... \wedge \text{d}\phi_{e_{2n}}.
\end{equation}
These coefficients depend only on the graphs; they need to be calculated once and for all to quantize any Poisson structure. We use the following recipe:
\begin{enumerate}
    \item Let $\mathcal{H}$ denote the upper half complex plane, $\text{Im}(z)>0$.
    \item Given $p,q \in \mathcal{H}$, define the function:
    \begin{equation*}
        \phi(p_1,p_2):=\frac{1}{2 \pi} \text{arg}\left(\frac{p_2-p_1}{p_2-\overline p_1}\right).
    \end{equation*}
   This function is best thought of as an angle, normalised by $2 \pi$. It can also be interpreted geometrically in terms of geodesics of $\mathcal{H}$ endowed with the Lobachevsky metric. This function can be extended continuously to the set the set of points $(p,q)$ with $Im(q)  \geq 0, \; p \neq q$. For a fuller discussion of this see \cite{Esposito}. \item Lets $\mathcal{H}_n$ be the set of pairwise distinct $n$-tuples of points in $\mathcal{H}$:
   \begin{equation*} 
       \mathcal{H}_n:=\{(p_1,...,p_n)|p_i \in \mathcal{H}, \; p_i \neq p_j \; \text{if} \; i \neq j\}.
   \end{equation*}
   This called the \emph{space of configurations}. It gets an orientation coming from the natural complex structure.
   \item Given a graph $\Gamma \in G_n$ and a configuration $(p_1,...,p_n)$, draw $\Gamma$ on $\mathbb{R}^2$ by assigning $p_k \in \mathcal{H}$ to the vertex $k$, $0$ to the vertex $L$, and $1$ to the vertex $R$.
   \item Each edge $e_k$ now defines an ordered pair of points $(p_i,p_j)$ and an angle $\phi_{e_k}:=\phi(p_i,p_j)$. 
\end{enumerate}
As the points $p_i$ move around $\mathcal{H}$, the $\phi_{e_i}$ define functions on $\mathcal{H}_n$ via projection. Then the integral in (\ref{ref10}) gives us the weight $\omega_{\Gamma}$; it is shown in \cite{Kontsevich} that it converges absolutely. Everything in (\ref{ref9}) is now defined. 

We now consider some simple weight integrals and star products. If $e_k$ is the edge  from vertex $p_i$ to vertex $p_j$, we can write $\phi_{e_k}$ alternatively as $\phi_{p_i\rightarrow p_j}$. We will first calculate the most basic wedge integral, corresponding to a graph with three vertices and two edges in $\mathcal{H}$. We refer to this setup as an `out-out' vertex. 
\begin{lemma}[General wedge integral] \label{lem2}
We have that:
\begin{equation} \label{ref11}
   \int_{x \in \mathcal{H}}\textnormal{d}\phi_{x \rightarrow y} \wedge \textnormal{d}\phi_{x \rightarrow z} = \frac{1}{2 \pi i} \textnormal{log} \left( \frac{y-\bar{z}}{z-\bar{y}} \right). 
\end{equation}

\end{lemma}
This lemma is stated in \cite{Panzer}, but the proof there is more complicated. Drawing heavy inspiration from the work of our supervisor, we have performed the calculation originally in a simpler way. A full proof is given in Appendix \ref{amend14}.
\begin{proof}
\emph{(Sketch)}. Apply the following algorithm: \begin{enumerate}
    \item Separate the integrand differential form into holomorphic and anti-holomorphic factors.
    \item Add on any antiholomorphic differentials to the holomorphic factor, to enable a single-valued primitive to be found.
    \item Apply Stokes' theorem to our single-valued integrand.
    \item Express the remaining contributions from Stokes' in terms of single valued functions depending on the remaining points. 
    \item Iterate this as many times as necessary.
\end{enumerate}
\end{proof}
Note we are not yet enforcing that two of the vertices lie on the real line, as in Kontsevich's actual construction. One can show the continuity of the integral (\ref{ref11}), so that in the limit as $y \rightarrow 0$, $z \rightarrow 1$, it implies that:
\begin{equation*}
   \int_{x \in \mathcal{H}}\textnormal{d}\phi_{x \rightarrow 0} \wedge \textnormal{d}\phi_{x \rightarrow 1} = \frac{1}{2}. 
\end{equation*}
In $G_1$ there are two graphs that differ by swapping the edges; they are exactly the `out-out' vertex graphs. Thus $\omega_1=\frac{1}{2}$ and by skew-symmetry of the differential form, $\omega_2=-\frac{1}{2}$. So $f \star g$ has contribution at order $\hbar$ given by:
\begin{equation*}
\frac{\hbar}{2}{\pi}^{ij}({\partial}_i f {\partial}_j g  -{\partial_j} f {\partial_i} g) = \hbar  {\pi}^{ij}\partial_i f \partial_j g,
\end{equation*}
and if we input $\frac{1}{2}\{.,.\}$, we get an induced  $B_1^-$ bracket of $\frac{1}{2}\{.,.\}$, exactly what was stated in Definition \ref{amend13}. 

Having performed this integral, we can now return to the Moyal product. 
\begin{prop}
Applying the Kontsevich formula (\ref{ref9}) to a constant Poisson structure gives us the Moyal product (\ref{ref13}).
\end{prop}
\begin{proof}
\begin{itemize}
\item Consider the constant Poisson structure with ${\pi}^{ij} \in \mathbb{R}$. Any graph with an edge ending in a vertex other than $L$ or $R$ gives contribution $0$, as there is a term ${\partial}_i {\pi}^{jk}=0$. 

\begin{figure} 
\includegraphics[width=4.5cm]{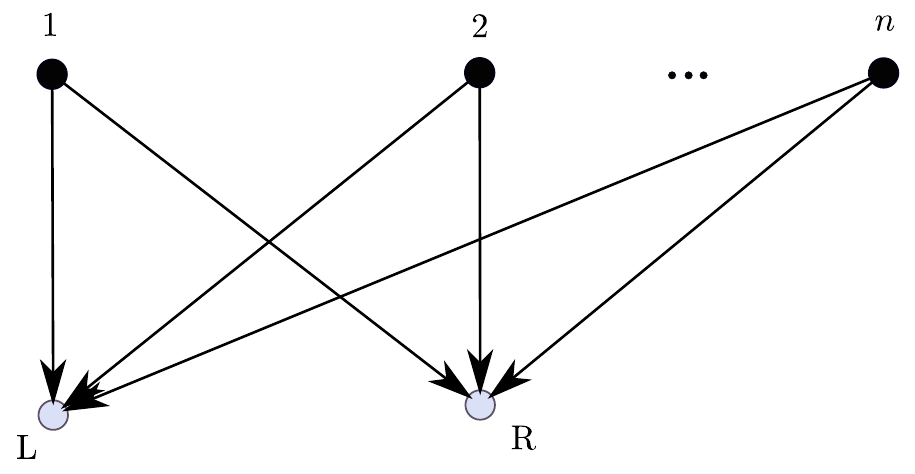}
\centering
\caption{A graph of order $n$ with non-zero contribution to the Moyal product.}
\centering
\label{fig:3}
\end{figure} 
\item The only non-zero contribution graphs have vertices with all edges going into $L$ or $R$. There are $2^n$ such graphs in $G_n$, differing by swapping pairs of edges from each vertex. 

\item Swapping the edges does not change the contribution to the star product, as we saw for ${\Gamma_1}, {\Gamma_2}$, because ${\pi}^{ij}=-{\pi}^{ji}$. 
\end{itemize}
\ \\
Let $\Gamma$ be the graph of this form such that the first labelled edge from each vertex goes into L. Because the contribution from all graphs is the same, at order $n$ we have:
\begin{equation*}
B_n(f,g) = \frac{2^n}{n!} {\omega_{\Gamma}} B_{\Gamma,\pi}(f,g) = \frac{2^n}{n!} {\omega}_{\Gamma} \left( \prod_{k=1}^{n} {\pi}^{i_k j_k} \right) \left( \prod_{k=1}^{n} \partial_{i_k} f \right) \left( \prod_{k=1}^{n} \partial_{j_k} g \right).
\end{equation*}
For $\omega_{\Gamma}$, the differential $2n$-form decomposes trivially, so that:
\begin{equation*}
\int_{\mathcal{H}^n} \text{d} \phi_{e_1}  \wedge...\wedge \text{d} \phi_{e_n}  = \left( \int_{x \in \mathcal{H}} \phi_{x \rightarrow 0} \wedge  \phi_{x \rightarrow 1}  \right)^n = \frac{1}{2^n}.
\end{equation*}
Putting this together, the $2^n$ in the numerator and denominator cancel so that:
\begin{equation*}
    f  \star  g = \sum_{n} \frac{{\hbar}^n}{n!} \left( \prod_{k=1}^{n} {\pi}^{i_k j_k} \right) \left( \prod_{k=1}^{n} \partial_{i_k} f \right) \left( \prod_{k=1}^{n} \partial_{j_k} g \right),
\end{equation*}
exactly what we want, i.e. equation (\ref{ref13}).

\end{proof}

We can see that breaking up graphs into smaller pieces is a good technique for dealing with more complicated graphs. If we first enforce $y=0,z=1$, the integral (\ref{ref11}) is considerably simpler: see, for example, the proof given in \cite{Esposito}, which uses an easy trick using the range of integration. However, we have presented a more involved calculation for general $y$ and $z$, because the algorithm generalises when one considers more involved weight integral calculations. 

We are always using the principle branch of $\text{log}(z)$ and our method means that the argument is never allowed to be negative. As an example, the argument in our answer, $\frac{y-\bar{z}}{z-\bar{y}}$, is never in $\mathbb{R}^-$. If one considers larger graphs, where further integration is be needed (for example of $y$ and $z$), the use of single-valued functions throughout means that one does not need to carry data about choices of log branches forward in the calculation. In the next chapter, we discuss the general algorithm of \cite{Panzer} for computing weight integrals. It is just the same as one presented above, generalised to a more sophisticated theoretical setting; running through Appendix \ref{amend14} therefore gives the reader some insight into how the results are achieved.

\section{Multiple zeta values}
To conclude, we explore the weight integrals in further depth, focusing in particular on the recent work \cite{Panzer}, of Erik Panzer, Peter Banks and Brent Pym. We also consider the connections of the subject to other areas of mathematics. This section draws on \cite{Zagier}, \cite{Waldschmidt}, \cite{MZVs} for background material.
\subsection{Periods}
We traditionally think of our number systems as falling in the hierarchy:
\begin{align*}
\mathbb{N} \subset  \mathbb{Z}  \subset  \mathbb{Q} \subset \overline{\mathbb{Q}}\\ \cap \: \quad \cap \\
\mathbb{R} \subset \mathbb{C}
\end{align*}
(Taken from \cite{Zagier}). There is, however, another class of numbers between $\overline{\mathbb{Q}}$ and $\mathbb{C}$, known as the \emph{periods}. Intuitively, these can be thought of as complex numbers expressible with a finite amount of algebraic information; they are a natural extension of $\overline{\mathbb{Q}}$. More precisely, we can state:
\begin{defn}[Period]
A period is a complex number, such that the both the real and imaginary parts are absolutely convergent integrals of rational functions with rational coefficients, over domains in $\mathbb{R}^n$ given by rational coefficient polynomial inequalities.
\end{defn}
A particular class of periods, the \emph{multiple zeta values}, or MZVs, is particularly relevant for the present discussion. Recall the standard definition of the \emph{Riemann zeta function} for $s\in \mathbb{C}$:
\begin{equation*} \label{ref43}
\zeta(s):=\sum_{n=1}^{\infty}n^{-s}.
\end{equation*}
This leads to:
\begin{defn}[Multiple zeta values]
Generalising the zeta values we have:
\begin{equation*}
 \zeta(s_1,...,s_k) := \sum_{0<n_1<...<n_k} \frac{1}{{n_1}^{s_1}...{n_k}^{s_k}},
\end{equation*}
for $s:=\{s_1,...,s_k\},s_i \in \mathbb{Z}^+$, $s_k>1$. All MZVs are convergent. 
\end{defn}
For a given MZV, we say that $|s|=s_1+...+s_k$ is the weight. There are many interesting relations between the MZVs; for a broader discussion see \cite{MZVs}.
\subsection{MZVs and weight integrals}
Consider the admissible graphs of Figure \ref{fig:2}. We have already calculated the weight of $(i)$ as $\frac{1}{2}$; considering the graphs in $(ii)$ from left to right, one can calculate the weights as: $\frac{1}{4}, \frac{1}{12}, -\frac{1}{12}, \frac{1}{24}$ respectively (The first is just a Moyal graph). We may think we see a simple pattern, but this idea is quickly dispelled by considering $\Gamma$ in Figure \ref{fig:4}. The corresponding weight is:
\begin{equation} \label{ref44}
\omega_{\Gamma}=-\frac{1}{6048}+\frac{9}{128}\frac{{\zeta(3)}^2}{{\pi}^6}.
\end{equation}
\begin{figure} 
\includegraphics[width=6cm]{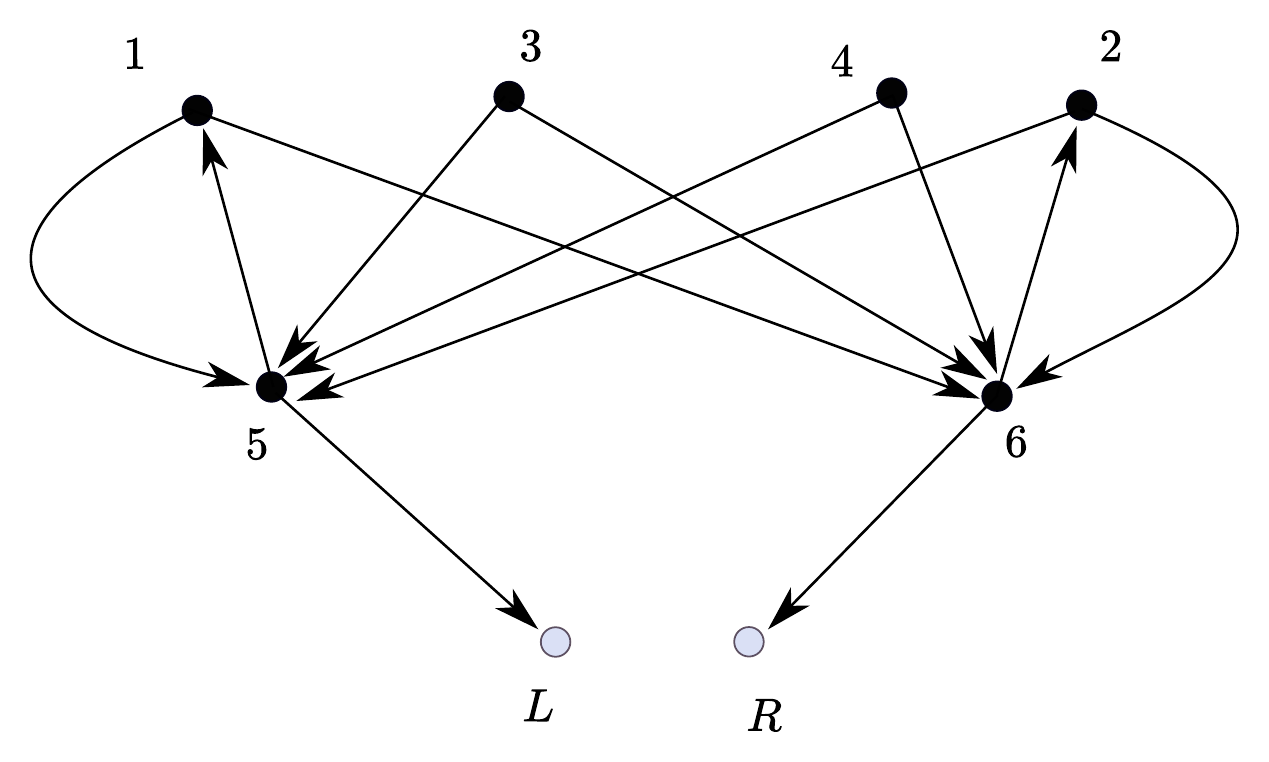}
\centering
\caption{An admissible graph of order 6.}
\label{fig:4}
\end{figure}
See \cite{Website}. The MZV $\zeta(3)$ has emerged from a weight integral calculation. In fact, Kontsevich conjectured himself that the weights are given by rational linear combinations of MZVs. This was recently proven in \cite{Panzer}:
\begin{thm}[Banks, Panzer, Pym, \cite{Panzer}]
Given a star product graph $\Gamma$ at order $\hbar^n$, the weight $\omega_{\Gamma}$ can be expressed as:
\begin{equation*}
 \omega_{\Gamma}=\sum_{|s|\leq n}\frac{\zeta(s)}{(i \pi)^{|s|}}\mathbb{Q}.
\end{equation*}
\end{thm}
Because $\omega_{\Gamma} \in \mathbb{R}$, the $i$ in the denominator implies that some of the rational coefficients must be zero. In fact, the authors of \cite{Panzer} go further, as the constructive proof gives a general algorithm for computing any $\omega_{\Gamma}$. 
\subsection{Polylogarithms}
To understand some ideas behind the proof, we introduce:
\begin{defn}[Polylogarithm] Define the $n^{th}$ polylogarithm function, $n \geq 1$ as the Chen iterated integral:
\begin{equation*}
Li_{n}(x) := \int_{0}^{x} \frac{\textnormal{d}z}{1-z} \left( \frac{\textnormal{d}z}{z}\right)^{n-1}.
\end{equation*}
This depends on a choice of path, but only up to homotopy equivalence: it is therefore a multi-valued function like a logarithm. By choosing the branch cut $(-\infty,0] \cup [1,\infty)$, we get a single-valued function. This is an improper integral, but by choosing a suitable limit we can find the value for paths with endpoints $0$ and $1$.
\end{defn}
There are differing conventions for ordering the differential forms in the integrand: we assume that the leftmost integral is performed first. This now leads to:
\begin{lemma} \label{lem5} Polylogarithms can be related to the Riemann zeta function with the identity:
\begin{equation*} 
Li_n(1)=\zeta(n).
\end{equation*}
\end{lemma}
\begin{proof}
By choosing the branch cut $(-1,0]$, we have a simply-connected domain $D=\{|x|<1\}\setminus(-1,0]$. So all paths in $D$ are homotopy equivalent and $Li_n(x)$ is independent of the path and single-valued. Now:
\begin{equation*}
\int_{0}^{x} \frac{\textnormal{d}z}{1-z} \left( \frac{\textnormal{d}z}{z}\right)^{n-1}=\int_{0}^{x} -\textnormal{log}(1-z) \left( \frac{\textnormal{d}z}{z}\right)^{n-1}.
\end{equation*}
Using the convergent Taylor series for $\textnormal{log}(1-z)$ in $|z|<1$, this is equal to:
\begin{equation*}
\int_{0}^{x}\left( z+ \frac{z^2}{2}+\frac{z^3}{3}+...\right) \left( \frac{\textnormal{d}z}{z}\right)^{n-1} =\int_{0}^{x}\left( z+ \frac{z^2}{2^2}+\frac{z^3}{3^2}+...\right) \left( \frac{\textnormal{d}z}{z}\right)^{n-2}
 =...=\sum_{i=1}^{\infty}\frac{x^i}{i^n}.
\end{equation*}
Taking the limit as $x\rightarrow 1$ we obtain $\zeta(n)$ as desired.
\end{proof}
Now, we can easily generalise this construction to:
\begin{defn}[Multiple polylogarithm] The single-variable multiple polylogarithm is given by the Chen iterated integral:
\begin{equation*}
Li_{(s_1,...,s_n)}(x):=\int_{0}^{x}\frac{\textnormal{d}z}{1-z}\left(\frac{\textnormal{d}z}{z} \right)^{s_1-1}\frac{\textnormal{d}z}{1-z}...\frac{\textnormal{d}z}{1-z}\left(\frac{\textnormal{d}z}{z} \right)^{s_n-1},
\end{equation*}
for $s_1,...,s_n \in \mathbb{Z^+}$. This is a multi-valued function on $\mathbb{C} \setminus \{0,1\}$.
\end{defn}
The relevance is immediate from:
\begin{lemma}
The MZVs can be expressed with the identity:
\begin{equation*} \label{lem4}
Li_{(s_1,...,s_n)}(1)= \sum_{n_1>...>n_k>0} \frac{1}{{n_1}^{s_1}...{n_k}^{s_k}}.
\end{equation*}
\end{lemma}
\begin{proof} \emph{(Sketch)}. The idea is exactly the same as the proof of Lemma \ref{lem5}, expanding the integral using Taylor series and integrating term-by-term. See \cite{MZVs}
for full details. 
\end{proof}
For a broader discussion of these functions, and their relations to MZVs, see \cite{Waldschmidt}. We note a final interesting fact about MZVs. There is a long-standing `folklore conjecture' that the odd integer values of the zeta function:
\begin{equation*}
\pi, \; \zeta(3), \; \zeta(5), \; \zeta(7),...
\end{equation*}
are algebraically independent \cite{Brown}. This suggests that certain Kontsevich weights, such as $\omega_{\Gamma}$ in (\ref{ref44}), are, in fact, transcendental.
\subsection{A constructive proof}
Recall the sketch proof of Lemma \ref{lem2}; we just used the algorithm of \cite{Panzer}. However, the authors generalise the ideas we presented there by finding real-analytic, single-valued primitives in terms of multiple polylogarithms rather than ordinary logs. Integrating out variables and applying (regularized) Stokes' theorem \cite{Panzer}, MZVs are eventually obtained exactly through Lemma \ref{lem4}. 

Importantly, we can now theoretically give a deformation quantization of a Poisson structure up to any desired $\hbar$ order. In fact, the authors of \cite{Panzer} have created open source software for this, available at \cite{Website}, where there is also a full list of weights up to ${\hbar}^6$. Prior work used various tricks for calculating weight integrals in special cases. The \emph{Bernoulli graphs} were treated previously in \cite{Kathotia}, where the coefficients were found to be easily expressible in terms of Bernoulli numbers. The software was checked for consistency against these calculations (and others like it), producing the same results. 

The significance of \cite{Panzer} is clear when one considers the large amount of previous work (see also \cite{merkulov}, \cite{Shoikhet}, \cite{VandenBergh}) that has focused on these calculations. Bringing this all together, the authors have, in a sense, `closed' a long-standing area of research.

\subsection{Final remarks}

\begin{figure}  
\includegraphics[width=8.5cm]{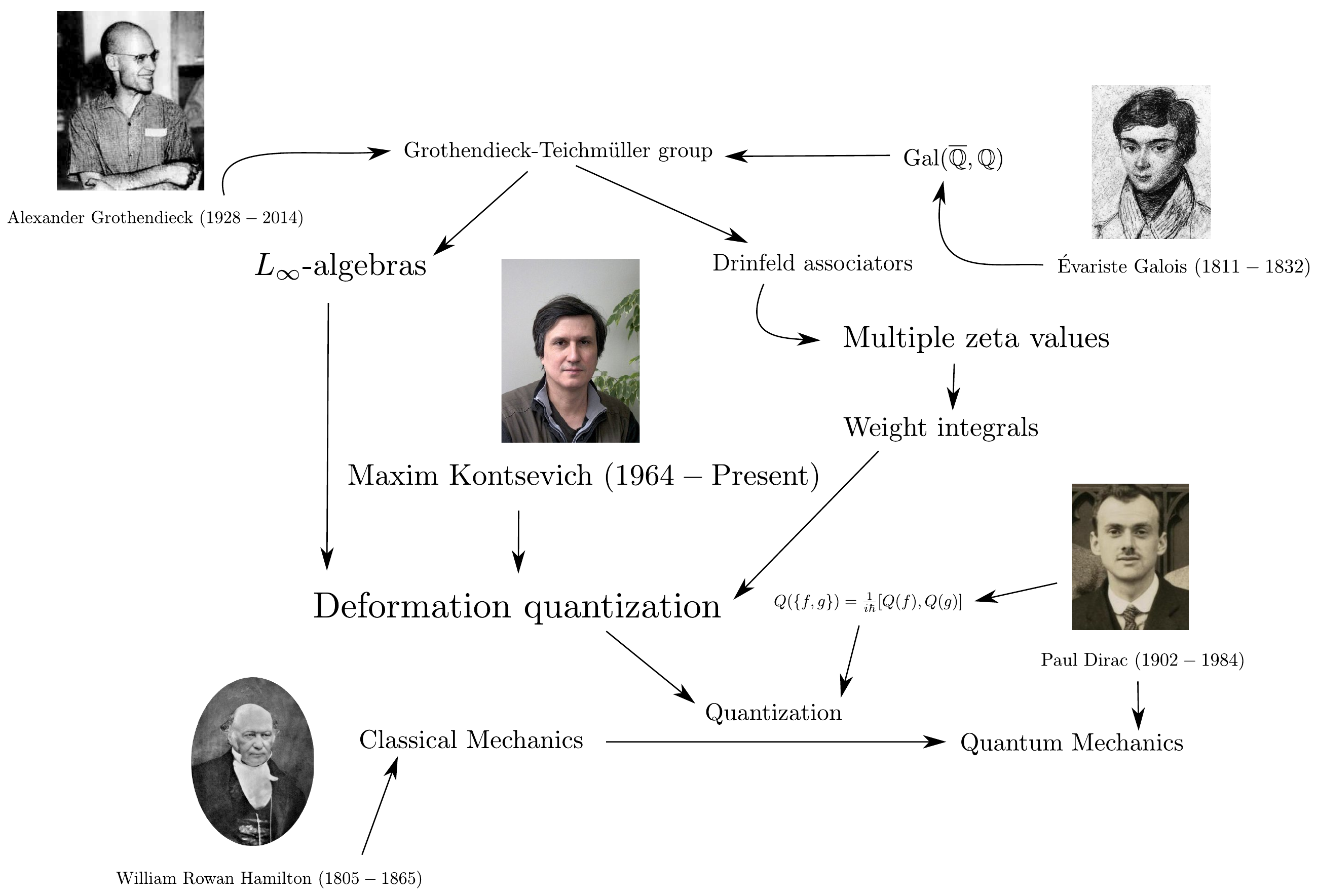}
\centering
\caption{Diagrammatic representation of the key ideas, figures and their relations.}
\label{fig5}
\centering
\end{figure} 

We began this essay with a simple, physically motivated aim: the quantization of classical mechanics. We saw that an early historical guess for this was given by Dirac, who suggested the quantization condition:
\begin{equation*}Q(\{f,g\})=\frac{1}{i\hbar}[Q(f),Q(g)].
\end{equation*} 
After introducing the idea of deformation quantization, there was finally a mysterious connection of this physics problem to MZVs, typically associated with number theory. This is not the only strange link. Recall that the Kontsevich formality theorem gives us an $L_{\infty}$-quasi-isomorphism:
\begin{equation*}
    U: T_{poly}(M)\rightarrow D_{poly}(M).
\end{equation*}
Consider the space of stable $L_{\infty}$-quasi-isomorphisms from $T_{poly}(M)$ to $D_{poly}(M)$ (see \cite{Dolgushev1}), denoted by:
\begin{equation*}\text{Maps}^{L_{\infty}}_{q.i.}(T_{poly}(M), D_{poly}(M)).\end{equation*} We can introduce a notion of homotopy equivalence between the maps in this set; then equivalent maps induce the same bijection of Maurer-Cartan sets, in the sense of Theorem \ref{ref26} (see \cite{Dolgushev}). 

Remarkably, the \emph{real Grothendieck-Teichm{\"u}ller group} $GT(\mathbb{R})$ now acts on the set of homotopy equivalence classes, $\pi_0 \text{Maps}^{L_{\infty}}_{q.i.}(T_{poly}(M), D_{poly}(M))$ (see \cite{Dolgushev}, \cite{Dolgushev1}, \cite{Grothendieck}). We know that there is a homomorphism from the \emph{absolute Galois group of} $\mathbb{Q}$, Gal$(\overline{\mathbb{Q}}/\mathbb{Q})$, to the \emph{profinite} $\widehat{GT}$ \cite{Drinfeld}, and therefore to the $p$-adic points of the \emph{algebraic} $GT$:
\begin{equation*}
\text{Gal}(\overline{\mathbb{Q}} / \mathbb{Q}) \rightarrowtail \widehat{GT} \rightarrow GT(\mathbb{Q}_p).
\end{equation*}
So Gal$(\overline{\mathbb{Q}}/\mathbb{Q})$ would act on the space of $p$-adic deformation quantizations, if such a notion made sense, another puzzling connection with number theory.

In fact, $GT(\mathbb{R})$ also acts on another construction, the space of \emph{Drinfeld associators}, Drin, which are formal noncommututative power series in two variables satisfying a complex set of algebraic relations. Drinfeld gives a specific associator, $KZ$, with MZVs as coefficients \cite{Drinfeld}; it is hard to escape the feeling that none of this is a coincidence:

\begin{equation*}
KZ  \in  \text{Drin} \; \reflectbox{$\lcirclearrowright$} \; GT \lcirclearrowright \pi_0 \text{Maps}^{L_{\infty}}_{q.i.}(T_{poly}(M), D_{poly}(M)) \; \reflectbox{$\in$} \; \text{$L_{\infty}$ Kont}.
\end{equation*}

In the introduction, we stated that deformation quantization has ambiguous physical applicability, because there is no Hilbert space to act on and convergence is ignored. Whilst this criticism may be fair in some sense, the strange and beautiful links of the subject to many disparate areas of mathematics and physics (see Figure \ref{fig5}) lead us to believe that it is not the end of the story. To conclude, we pose three (possibly very naive) questions.
\begin{enumerate}
    \item How can deformation quantization be used for Hilbert space quantization?
    \item What are the underlying relations between deformation quantization and other mathematical structures where MZVs appear? 
    \item Are arithmetic geometry and physics related on some fundamental level?
   
\end{enumerate}

    \appendix
    \section{Proof of Lemma \ref{lem2}} \label{amend14}
    \begin{proof}
    Firstly,
\begin{equation*}
\textnormal{d} \phi_{x\rightarrow y} =\frac{1}{2\pi}\text{d}\text{arg}\left(\frac{y-x}{y-\bar{x}} \right)= \frac{1}{4 \pi i}\text{d}\text{log}\left(\frac{(y-x)(\bar{y}-x)}{(y-\bar{x})(\bar{y}-\bar{x})} \right).
\end{equation*}
We want to calculate:
\begin{equation*}
I = \int_{x \in \mathcal{H}}\textnormal{d}\phi_{x \rightarrow y} \wedge \textnormal{d}\phi_{x \rightarrow z} \\\
= \frac{1}{(4 \pi  i)^2} \int_{x \in \mathcal{H}} \text{d}\text{log}\left(\frac{(y-x)(\bar{y}-x)}{(y-\bar{x})(\bar{y}-\bar{x})} \right) \wedge \text{d}\text{log}\left(\frac{(z-x)(\bar{z}-x)}{(z-\bar{x})(\bar{z}-\bar{x})} \right).
\end{equation*}
Consider the integrand:
\begin{multline*}
    \text{d}\text{log}\left(\frac{(y-x)(\bar{y}-x)}{(y-\bar{x})(\bar{y}-\bar{x})} \right) \wedge \text{d}\text{log}\left(\frac{(z-x)(\bar{z}-x)}{(z-\bar{x})(\bar{z}-\bar{x})} \right) \\\ =\text{d}\text{log}((z-\bar{x})(\bar{z}-\bar{x})) \wedge \text{d}\text{log}((y-x)(\bar{y}-x))
    -\text{d}\text{log}((y-\bar{x})(\bar{y}-\bar{x})) \wedge \text{d}\text{log}((z-x)(\bar{z}-x)).
\end{multline*}
We want a single-valued primitive $\eta$, which we can now give:
\begin{multline*}
    \eta =
    \text{log}\left((z-\bar{x})(\bar{z}-\bar{x})(\bar{z}-x)(z-x) \right)  \text{d}\text{log}((y-x)(\bar{y}-x))\\\
    -\text{log}\left((y-\bar{x})(\bar{y}-\bar{x})(\bar{y}-x)(y-x) \right)  \text{d}\text{log}((z-x)(\bar{z}-x)).
\end{multline*}
We use the branch of the argument function from $-\pi$ to $\pi$. The integral we want is:
\begin{equation*}
I = \frac{1}{(4 \pi  i)^2} \int_{x \in \mathcal{H}} \text{d} \eta .  
\end{equation*}
Now apply Stokes' theorem. We use the curve shown in Figure \ref{figA2}, consisting of the segment $[-R,R]$, a semi-circular arc $\Gamma$ of radius $R$, and little circles of radius $r$ inside that cut out the problematic points. Now:
\begin{equation*}
     \lim_{r\to 0} \int_{\gamma_1}
     \text{log}\left((z-\bar{x})(\bar{z}-\bar{x})(\bar{z}-x)(z-x) \right)  \text{d}\text{log}((y-x)(\bar{y}-x))   \\\ = 2 \pi i  \text{log}\left( (z-\bar{y})(\bar{z}-\bar{y})(\bar{z}-y)(z-y) \right) . 
\end{equation*}
and also:
\begin{equation*}
   \lim_{r\to 0} \int_{\gamma_2}
   \text{log}\left((y-\bar{x})(\bar{y}-\bar{x})(\bar{y}-x)(y-x) \right)  \text{d}\text{log}((z-x)(\bar{z}-x)) \\\ = 2 \pi i  \text{log}\left( (y-\bar{z})(\bar{y}-\bar{z})(\bar{y}-z)(y-z) \right).
\end{equation*}
\begin{figure} 
\includegraphics[width=6cm]{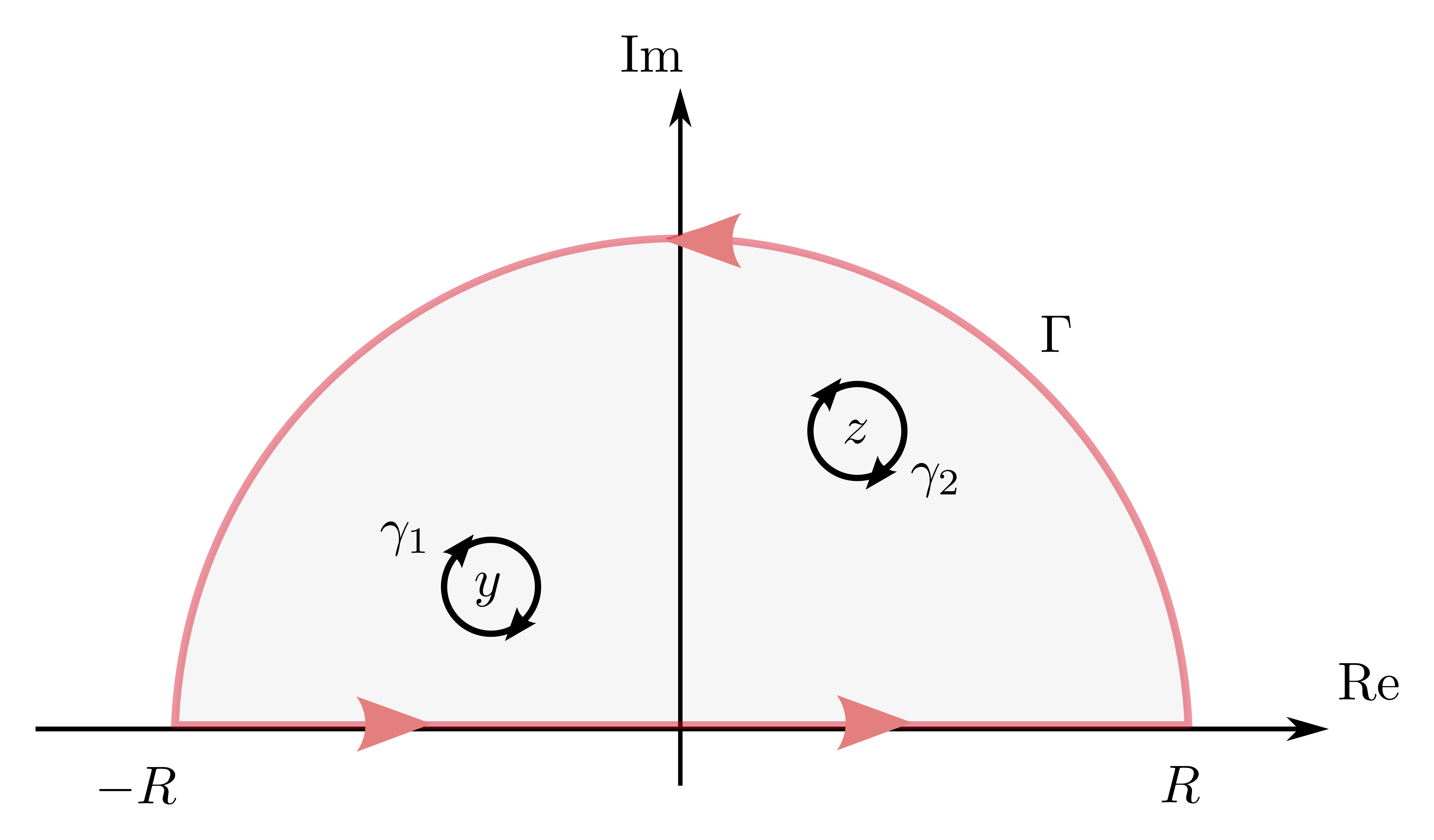}
\centering
\caption{The boundary curve  $\Gamma \cup [-R,R]\cup \gamma_1 \cup \gamma_2$.}
\centering
\label{figA2}
\end{figure}
As $r\rightarrow 0$, the first term integrated around $\gamma_2$ goes to $0$, and the second term integrated around $\gamma_1$ also goes to $0$ (essentially because $r\text{log}(r) \rightarrow 0 \; \text{as} \; r \rightarrow 0$). Now summing these up gives total contribution $0$. We now calculate the integral over $\Gamma$ and the integral over $[-R,R]$, in the limit as $R\rightarrow \infty$. Firstly:
\begin{multline*}
    \int_{\Gamma}\eta=
    \int_{\Gamma}  
    \text{log}\left((z-\bar{x})(\bar{z}-\bar{x})(\bar{z}-x)(z-x) \right)  \text{d}\text{log}((y-x)(\bar{y}-x))\\\
    -\text{log}\left((y-\bar{x})(\bar{y}-\bar{x})(\bar{y}-x)(y-x) \right)  \text{d}\text{log}((z-x)(\bar{z}-x)).
\end{multline*}
We calculate the first term first:
\begin{equation*}
A=\int_{\Gamma}\left( \text{log}({\bar{x}}^2 x^2)+\text{log} \left(\left(\frac{z}{\bar{x}}-1\right)\left(\frac{\bar{z}}{\bar{x}}-1\right)\left(\frac{\bar{z}}{x}-1\right)\left(\frac{z}{x}-1\right) \right) \right) \text{d}\text{log}((y-x)(\bar{y}-x)).
    \end{equation*}
Now as $R \rightarrow \infty$:
\begin{equation*}
\text{d}\text{log}((x-y)(x-\bar{y})) = \frac{\text{d}x}{x-y}+\frac{\text{d}x}{x-\bar{y}} = \frac{\text{d}x}{x}\left(\frac{1}{1-\frac{y}{x}} +\frac{1}{1-\frac{\bar{y}}{x}} \right) \rightarrow 2 \frac{\text{d}x}{x}.
\end{equation*}
Also:
\begin{equation*}
\text{log}({\bar{x}}^2 x^2)+\text{log}\left(\left(\frac{z}{\bar{x}}-1\right)\left(\frac{\bar{z}}{\bar{x}}-1\right)\left(\frac{\bar{z}}{x}-1\right)\left(\frac{z}{x}-1\right) \right) \rightarrow 2\text{log}(x\bar{x}).
\end{equation*}
In the limit, there is no $z$ dependence in the integrand of the first term, and the integrand of the second term tends to the same limit.  So these just cancel and we have a total contribution of $0$. For the integral over $\mathbb{R}$, we have:
\begin{multline*}
    \int_{\mathbb{R}}\eta=
    \int_{\mathbb{R}}  
    \text{log}\left((z-\bar{x})(\bar{z}-\bar{x})(\bar{z}-x)(z-x) \right)  \left( \frac{\text{d}x}{x-y}+\frac{\text{d}x}{x-\bar{y}} \right)\\\
    -\text{log}\left((y-\bar{x})(\bar{y}-\bar{x})(\bar{y}-x)(y-x) \right) \left( \frac{\text{d}x}{x-z}+\frac{\text{d}x}{x-\bar{z}} \right).
\end{multline*}
\begin{figure} 
\includegraphics[width=9cm]{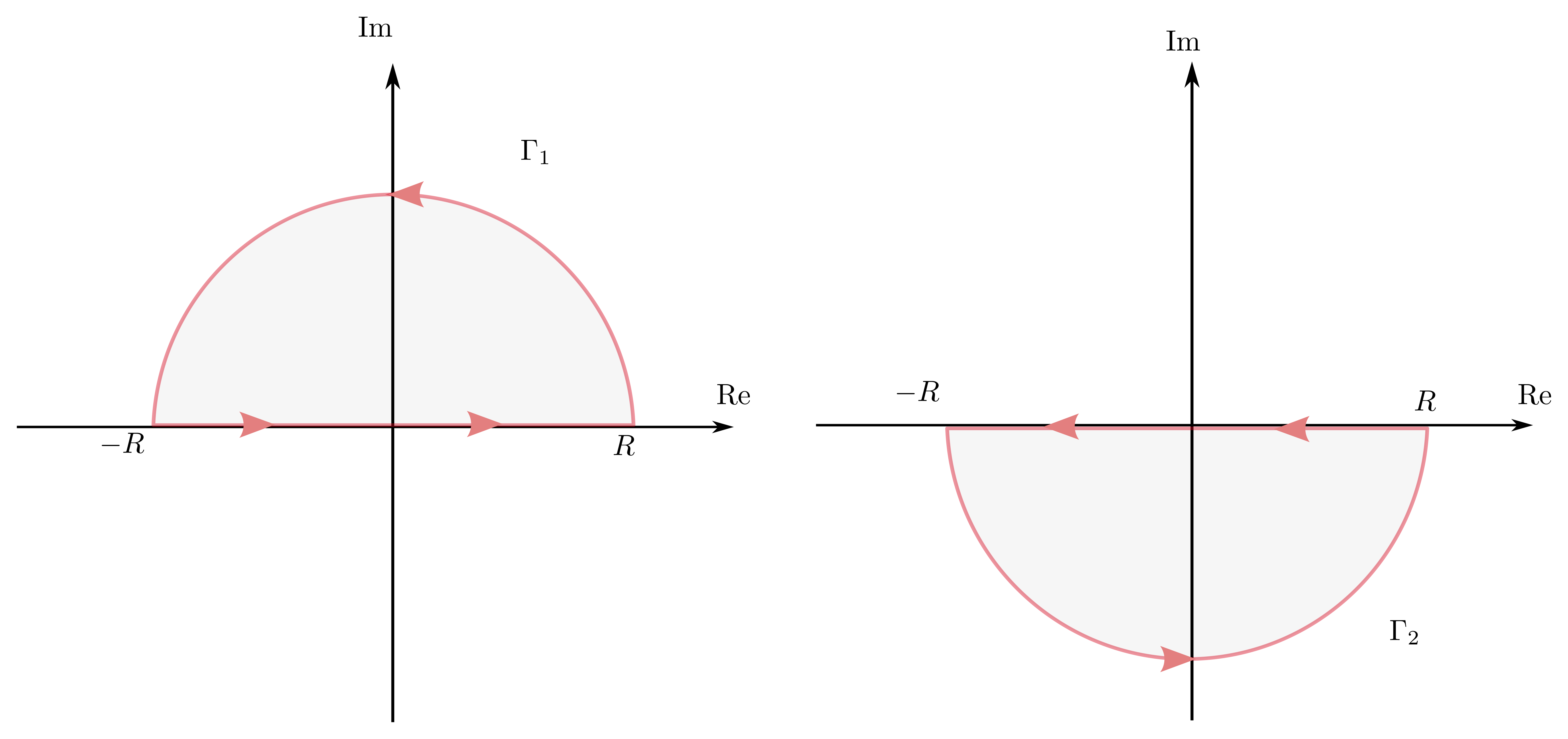}
\centering
\caption{The contours $\Gamma_1$ and $\Gamma_2$}
\label{figA1}
\centering
\end{figure}
Take two contours $\Gamma_1, \Gamma_2$, shown in Figure \ref{figA1}. $\Gamma_1$ consists of an upper semi-circular arc, and $[-R,R]$, as before. $\Gamma_2$ consists of a lower semi-circular arc and the oppositely oriented line segment. Because $x=\bar{x}$ on $\mathbb{R}$, as $R\rightarrow \infty$ we need to calculate:
\begin{equation*}
 \int_{\mathbb{R}}  
    \text{log}\left((z-x)^2(\bar{z}-x)^2) \right)  
    \text{d}\text{log}((y-x)(\bar{y}-x))
    -\text{log}\left((y-x)^2(\bar{y}-x)^2 \right) \text{d}\text{log}((z-x)(\bar{z}-x)).
\end{equation*}
Now along the real line, we can use the typical identities:
\begin{multline*}
\text{log}\left( (z-x)^2(\bar{z}-x)^2 \right)= \text{log}\left( (z-x)^2\right)+\text{log}  \left( (\bar{z}-x)^2 \right), \\ \text{log}\left( (y-x)^2(\bar{y}-x)^2 \right)= \text{log}\left( (y-x)^2\right)+\text{log}  \left( (\bar{y}-x)^2 \right) . 
\end{multline*}
We can use $\text{log}((\bar{z}-x)^2(z-x)^2)= 2 (\text{log}(\bar{z}-x)+\text{log}(z-x))$, because of the choice of the branch and as $z$, $\bar{z}$ are complex conjugates. The same holds for the $y, \bar{y}$ terms.

Now, in our integrand, two of the terms extend holomorphically to the upper half plane; the other two terms extend holomorphically to the lower half plane. We treat these pairs of terms separately. Firstly:
\begin{equation*}
\int_{\Gamma_1} 2\text{log}(\bar{z}-x)\text{d}\text{log}((y-x)(\bar{y}-x))-2\text{log}(\bar{y}-x)\text{d}\text{log}((z-x)(\bar{z}-x)),
\end{equation*}
can just be evaluated by taking the residues at $y,z$, in the upper half plane. We get the residue contributions:
\begin{equation*}
4 \pi i \left( \text{log}(\bar{z}-y)-\text{log}(\bar{y}-z)\right)= 4 \pi i  \text{log}\left( \frac{y-\bar{z}}{z-\bar{y}}\right).
\end{equation*}
Secondly:
\begin{equation*}
  \int_{\Gamma_2}  2\text{log}(z-x)\text{d}\text{log}((y-x)(\bar{y}-x))-2\text{log}(y-x)\text{d}\text{log}((z-x)(\bar{z}-x)),
\end{equation*}
can be evaluated by taking the residues at $\bar{y},\bar{z}$, in the lower half plane. We get the residue contributions:
\begin{equation*}
4 \pi i \left(   \text{log}(z-\bar{y})-  \text{log}(y-\bar{z})  \right)=4 \pi i  \text{log}\left( \frac{z-\bar{y}}{y-\bar{z}}\right).
\end{equation*}
The integrals along the upper and lower semicircles vanish as before for $R \rightarrow \infty$. The total integral along $\mathbb{R}$ is therefore:
\begin{equation*}
\int_{\mathbb{R}} \eta = 8 \pi i \text{log} \left( \frac{y-\bar{z}}{z-\bar{y}}\right).
\end{equation*}
Putting this together:
\begin{equation*}
    I=\frac{1}{2 \pi i}\text{log}\left( \frac{y-\bar{z}}{z-\bar{y}} \right),
\end{equation*}
as desired.
\end{proof}

\bibliographystyle{plain}
\bibliography{ref}
\end{document}